\documentclass[10pt,a4paper]{article}
\usepackage[utf8]{inputenc}
\usepackage[T1]{fontenc}
\usepackage[margin=2.54cm]{geometry}
\usepackage{amssymb}
\usepackage{graphicx}
\usepackage{hyperref}
\usepackage{amsthm}
\usepackage{setspace}
\usepackage{amsmath}
\usepackage{booktabs}
\newtheorem{theorem}{Theorem}[section]
\newtheorem{definition}[theorem]{Definition}
\newtheorem{corollary}[theorem]{Corollary}
\newtheorem{lemma}[theorem]{Lemma}
\newtheorem{proposition}[theorem]{Proposition}
\newtheorem{example}[theorem]{Example}
\numberwithin{equation}{section}
\begin{document}
	\title{Principal hierarchy of Infinite-dimensional Frobenius Manifold
		Underlying the Extended Kadomtsev-Petviashvili Hierarchy}
	\author{Shilin Ma$^{a}$}
	\date{\textit{$^{a}$School of Mathematical Science, Tsinghua University, Beijing, 100084, China}}
	\maketitle
	\begin{abstract}
		In this paper, we construct the principal hierarchy of the infinite-dimensional Frobenius manifold
		underlying the extended Kadomtsev-Petviashvili hierarchy. We show that this hierarchy
		serves as an extension of the genus zero Whitham hierarchy with one movable pole.
		
	\end{abstract}
	\section{Introduction
	}
	The concept of Frobenius manifold, introduced by Dubrovin in [1], provides a geometric formulation
	of the associativity equations of two-dimensional topological field theory (2D TFT). This concept has important applications in several branches of mathematical physics, including Gromov-Witten theory,
	singularity theory, integrable systems ect., see for example, [2, 3, 4, 5, 6] and references therein.
	
	Every Frobenius manifold is associated with a principal hierarchy of Hamiltonian equations of
	hydrodynamic type, in which the unknown functions depend on one scalar spatial variable and some
	time variables. In the case where the corresponding Frobenius manifold is semisimple, the principal
	hierarchy can be extended to a full hierarchy. The tau function selected by the string equation of this
	hierarchy yields the partition function of the two-dimensional topological field theory. This important
	connection between Frobenius manifolds and 2D TFT has led to significant progress in the study of
	both topics, and has deepened our understanding of the geometric and algebraic structures that arise
	in these contexts. Please see [4, 7, 8, 9] for details.

	Trials to extend the above programme to integrable systems with two spatial dimensions started in
	recent years. The first significant progress in this direction was made by Carlet, Dubrovin and Mertens
	[10] in consideration of the dispersionless 2D Toda lattice hierarchy. More exactly, they discovered an
	infinite-dimensional Frobenius manifold structure on a space of pairs of certain meromorphic functions
	with single poles at the origin and at infinity. Building on this approach, several more infinite-dimensional
	Frobenius manifolds have been constructed by considering pairs of meromorphic functions
	with higher-order poles at the origin and at infinity, and such manifolds underlie the two-component
	Kadomtsev-Petviashvili hierarchy of Type B [11], the $(N,M)$-type Toda lattice hierarchy [12] and the
	two-component extension of the Kadomtsev-Petviashvili (KP) hierarchy (or the genus zero Whitham hierarchy
	with one movable pole, without the logarithm flow) [13]. In a different approach, Raimondo [14]
	constructed a Frobenius manifold structure on the space of Schwartz functions for the dispersionless
	KP hierarchy. Additionally, Szablikowski [15] proposed a scheme for constructing Frobenius manifolds
	using the Rota-Baxter identity and a modified version of the Yang-Baxter equation for the classical
	r-matrix.
	
	In order to obtain the principal hierarchy of an infinite-dimensional Frobenius manifold, one must
	find the solution in a normal form (Levelt form) of the associated deformed flatness equations. In
	a recent paper, Carlet and Mertens [16] obtain the principal hierarchy of the infinite-dimensional
	Frobenius manifold that was constructed in a previous reference. This hierarchy, which is characterized
	by an extra ”logarithmic” flow, serves as an extension of the dispersionless 2D Toda hierarchy. Inspired by their groundbreaking work, our present study aims to constructing the principal hierarchy of the infinite-dimensional Frobenius manifold that underlies the two-component extension of the KP hierarchy, employing a distinctive approach. Furthermore, we demonstrate that this hierarchy serves as a natural extension of the genus zero Whitham hierarchy with two marked points.
	
	To provide a more precise statement of our results, it is necessary to review the definitions of
	Frobenius manifold and associated principal hierarchy following [1, 4].

	A Frobenius algebra $(A, \circ, e,\langle\ ,\ \rangle)$ is defined as a commutative associative algebra $(A, \circ)$ with a unity element $e$ and a nondegenerate symmetric invariant bilinear form $\langle$\ ,\ $\rangle$. A Frobenius manifold of charge $d$ is an $n$-dimensional manifold $M$ where each tangent space $T_v M$ is equipped with a Frobenius algebra structure $\left(A_v=T_v M, \circ, e,\langle\ ,\ \rangle\right)$ that depends smoothly on $v \in M$, such that the following three axioms are satisfied:
	\begin{enumerate}
		\item The bilinear form $\langle\ ,\ \rangle$ is a flat metric on $M$, and the unity vector field $e$ satisfies $\nabla e=0$, where $\nabla$ denotes the Levi-Civita connection for the flat metric;
		\item Let $c$ be a 3-tensor defined as $c(X, Y, Z):=\langle X \circ Y, Z\rangle$ where $X, Y, Z \in T_v M$. Then the 4-tensor $\left(\nabla_W c\right)(X, Y, Z)$ is symmectric in $X, Y, Z, W \in T_v M$;
		\item There exists a vector field $E$, called the Euler vector field, which satisfies $\nabla \nabla E=0$ and the following conditions are satisfied for any vector fields $X$ and $Y$ on $M$:
		$$
		\begin{aligned}
			& {[E, X \circ Y]-[E, X] \circ Y-X \circ[E, Y]=X \circ Y,} \\
			& E(\langle X, Y\rangle)-\langle[E, X], Y\rangle-\langle X,[E, Y]\rangle=(2-d)\langle X, Y\rangle .
		\end{aligned}
		$$
	\end{enumerate}
	
	On an n-dimensional Frobenius manifold $M$, it is possible to select a set of flat coordinates $t=$ $\left(t^1, \cdots, t^n\right)$ in which the unity vector field is expressed as $e=\frac{\partial}{\partial t^1}$. In these chosen coordinates, we have a constant non-degenerate $n \times n$ matrix defined as follows:
	$$
	\eta_{\alpha \beta}=\left\langle\frac{\partial}{\partial t^\alpha}, \frac{\partial}{\partial t^\beta}\right\rangle, \quad \alpha, \beta=1, \cdots, n,
	$$
	and its inverse is denoted by $\eta^{\alpha \beta}$. The matrices $\eta_{\alpha \beta}$ and $\eta^{\alpha \beta}$ will be used to lower and to lift indices, respectively, and summations over repeated Greek indices are assumed. Let
	$$
	c_{\alpha \beta \gamma}=c\left(\frac{\partial}{\partial t^\alpha}, \frac{\partial}{\partial t^\beta}, \frac{\partial}{\partial t^\gamma}\right), \quad \alpha, \beta, \gamma=1, \cdots, n,
	$$
	then the product of the Frobenius algebra $T_v M$ can be expressed as
	$$
	\frac{\partial}{\partial t^\alpha} \circ \frac{\partial}{\partial t^\beta}=c_{\alpha \beta}^\gamma \frac{\partial}{\partial t^\gamma},
	$$
	where the coefficients $c_{\alpha \beta}^\gamma$ are given by
	$$
	c_{\alpha \beta}^\gamma=\eta^{\gamma \epsilon} c_{\epsilon \alpha \beta}
	$$
	that satisfy
	$$
	c_{1 \alpha}^\beta=\delta_\alpha^\beta, \quad c_{\alpha \beta}^\epsilon c_{\epsilon \gamma}^\sigma=c_{\alpha \gamma}^\epsilon c_{\epsilon \beta}^\sigma .
	$$
	Moreover, locally, there exists a smooth function $F(t)$ such that:
	$$
	\begin{aligned}
		c_{\alpha \beta \gamma} & =\frac{\partial^3 F}{\partial t^\alpha \partial t^\beta \partial t^\gamma}, \\
		\operatorname{Lie}_E F & =(3-d) F+\text { quadratic terms in t. }
	\end{aligned}
	$$
	In other words, the function $F$ solves the Witten-Dijkgraaf-Verlinde-Verlinde(WDVV) equation:
	$$
	\frac{\partial^3 F}{\partial t^\alpha \partial t^\beta, \partial t^\gamma} \eta^{\gamma \epsilon} \frac{\partial^3 F}{\partial t^\epsilon \partial t^\sigma \partial t^\mu}=\frac{\partial^3 F}{\partial t^\alpha \partial t^\sigma, \partial t^\gamma} \eta^{\gamma \epsilon} \frac{\partial^3 F}{\partial t^\epsilon \partial t^\beta \partial t^\mu},
	$$
	and its third-order derivatives $c_{\alpha \beta \gamma}$ are referred to as the 3-point correlator functions in topological field theory. Conversely, given a solution $F$ of the above WDVV equation, one can reconstruct the structure of a Frobenius manifold.
	
	For the Frobenius manifold $M$, its cotangent space $T_v^* M$ also carries a Frobenius algebra structure. This structure includes an invariant bilinear form and a product defined as follows:
	$$
	\left\langle d t^\alpha, d t^\beta\right\rangle=\eta^{\alpha \beta}, \quad d t^\alpha \circ d t^\beta=\eta^{\alpha \epsilon} c_{\epsilon \gamma}^\beta .
	$$
	Let
	\begin{equation}\label{inter}
	g^{\alpha \beta}=i_E\left(d t^\alpha \circ d t^\beta\right),
	\end{equation}
	then $\left(d t^\alpha, d t^\beta\right):=g^{\alpha \beta}$ defines a symmectric bilinear form, called the intersection form, on $T_v^* M$.
	
	The two bilinear forms $\eta^{\alpha \beta}$ and $g^{\alpha \beta}$ on $T^* M$ compose a pencil $g^{\alpha \beta}+\epsilon \eta^{\alpha \beta}$ of flat metrics parametrized by $\epsilon$. Consequently, they induce a bi-hamiltonian structure $\{\ ,\ \}_2+\epsilon\{\ ,\ \}_1$ of hydrodynamic type on the loop space $\left\{S^1 \rightarrow M\right\}$.
	
	The deformed flat connection on $M$ defined as
\begin{equation}\label{defcon1}
	\widetilde{\nabla}_X Y=\nabla_X Y+z X \circ Y, \quad X, Y \in \operatorname{Vect}(M)
	\end{equation}
	can be extended to a flat affine connection on $M \times \mathbb{C}^*$ by the following definitions:
	\begin{equation}\label{defcon2}
	\tilde{\nabla}_X \frac{d}{d z}=0, \quad \tilde{\nabla}_{\frac{d}{d z}} \frac{d}{d z}=0, \quad \tilde{\nabla}_{\frac{d}{d z}} X=\partial_z X+E \circ X-\frac{1}{z} \mathcal{V}(X), 
	\end{equation}
	where $X$ are vector field on $M \times \mathbb{C}^*$ with zero components along $\mathbb{C}$, and
	$$
	\mathcal{V}(X):=\frac{2-d}{2} X-\nabla_X E.
	$$
	There exists a system of deformed flat coordinates $\tilde{v}_1(t, z), \cdots, \tilde{v}_n(t, z)$ of $M$ that can be expressed as:
	\begin{equation}\label{defsol}
		\left(\tilde{v}_1(t, z), \cdots, \tilde{v}_n(t, z)\right)=\left(\theta_1(t, z), \cdots, \theta_n(t, z)\right) z^\mu z^R
	\end{equation}
	and such that
	$$
	\xi_\alpha=\frac{\partial \tilde{v}_\alpha}{\partial t^\beta} d t^\beta, \quad \alpha=1, \cdots, n, \quad \xi_{n+1}=d z
	$$
	yield a basis of solutions of the system $\widetilde{\nabla} \xi=0$. Here $\mu=\operatorname{diag}\left(\mu_1, \cdots, \mu_n\right)$ is a diagonal matrix defined by
	\begin{equation}\label{bigmu}
		\mu_\alpha \frac{\partial}{\partial t^\alpha}=\mathcal{V}\left(\frac{\partial}{\partial t^\alpha}\right), \quad \alpha=1, \cdots, n,
	\end{equation}
	and $R=R_1+\cdots+R_m$ is a constant nilpotnet matrix satisfing
	$$
	\begin{aligned}
		& \left(R_s\right)_\beta^\alpha=0 \text { if } \mu_\alpha-\mu_\beta=s, \\
		& \left(R_s\right)_\alpha^\gamma \eta_{\gamma \beta}=(-1)^{s+1}\left(R_s\right)_\beta^\gamma \eta_{\gamma \alpha} .
	\end{aligned}
	$$
	The fumctions $\theta_\alpha(t, z)$ are analytic near $z=0$ and can be expanded as follows:
	\begin{equation}
		\theta_\alpha(t, z)=\sum_{p \geq 0} \theta_{\alpha, p}(t) z^p, \quad \alpha=1, \cdots, n .
	\end{equation}
	Then, the equations \eqref{defcon1}, \eqref{defcon2} for the deformed flat coordinates \eqref{defsol} are equivalent to
	\begin{equation}\label{princon1}
	\frac{\partial^2 \theta_{\alpha, p+1}(t)}{\partial t^\beta \partial t^\gamma}=c_{\beta \gamma}^\epsilon(t) \frac{\partial \theta_{\alpha, p}(t)}{\partial t^\epsilon},
	\end{equation}
	and
	\begin{equation}\label{princon2}
	\operatorname{Lie}_E\left(\partial_\beta \theta_{\alpha, p}(t)\right)=\left(p+\mu_\alpha+\mu_\beta\right)\left(\partial_\beta \theta_{\alpha, p}(t)\right)+\partial_\beta \sum_{s=1}^p \theta_{\epsilon, s}(t)\left(R_{p-s}\right)_\alpha^\epsilon .
	\end{equation}
Morever, we impose the following initial condition\footnote{In the referenced work [1], an additional condition $\left\langle\nabla \theta_\alpha(t, z), \nabla \theta_\beta(t,-z)\right\rangle=\eta_{\alpha \beta}$ was imposed. Nevertheless, since this condition does not substantially affect the properties of the principal hierarchy, we will exclude it in the present paper for the sake of computational simplicity.} on $\theta_{\alpha, 0}(t)$ :
\begin{equation}\label{princon3}
\theta_{\alpha, 0}(t)=\eta_{\alpha \beta} t^\beta.
\end{equation}

Given a system of solutions $\left\{\theta_{\alpha, p}\right\}$ of the equations \eqref{princon1}-\eqref{princon3}, the principal hierarchy associated with $M$ is defined as the following system of Hamiltonian equations on the loop space $\left\{S^1 \rightarrow M\right\}$ :
\begin{equation}\label{prin}
\frac{\partial t^\gamma}{\partial T^{\alpha, p}}=\left\{t^\gamma(x), \int \theta_{\alpha, p+1} (t) d x\right\}_1:=\eta^{\gamma \beta} \frac{\partial}{\partial x}\left(\frac{\partial \theta_{\alpha, p+1} (t)}{\partial t^\beta}\right), \quad \alpha, \beta=1,2, \cdots, n, \ p \geq 0.
\end{equation}
These commutating flows are tau-symmectric, meaning that:
$$
\frac{\partial \theta_{\alpha, p}(t)}{\partial T^{\beta, q}}=\frac{\partial \theta_{\beta, q}(t)}{\partial T^{\alpha, p}}, \quad \alpha, \beta=1,2, \cdots, n, \ p, q \geq 0 .
$$
Furthermore, these flows can be expressed in a bihamiltonian recursion form as:
$$
\mathcal{R} \frac{\partial}{\partial T^{\alpha, p-1}}=\frac{\partial}{\partial T^{\alpha, p}}\left(p+\mu_\alpha+\frac{1}{2}\right)+\sum_{s=0}^p \frac{\partial}{\partial T^{\epsilon, p-s}}\left(R_{p-s}\right)_s^\epsilon,
$$
where $\mathcal{R}=\{\ ,\ \}_2 \cdot\{\ ,\ \}_1^{-1}$.

The main results of the present paper are summarized as follows.

The paper focuses on the infinite-dimensional Frobenius manifold denoted by $\mathcal{M}_{m, n}^s$ [13], which serves as the underlying structure for the two-component extension of the KP hierarchy defined in [17]. This manifold consists of a pair of Laurent series given by
$$
(a(z), \hat{a}(z))=\left(z^m+\sum_{i \leq m-2} a_i(z-\varphi)^i, \sum_{i \geq-n} \hat{a}_i(z-\varphi)^i\right),
$$
subject to certain additional conditions (explained in section 2.1 below). 
The set $\mathbf{t} \cup \mathbf{h} \cup \hat{\mathbf{h}}=\left\{t_i\right\}_{i \in \mathbb{Z}} \cup$ $\left\{h_j\right\}_{1 \leq j \leq m-1} \cup\{\hat{h}_k\}_{0 \leq k \leq n}$
comprises flat coordinates for the metric $\eta$ associated with $\mathcal{M}_{m, n}^s$.
\begin{theorem}\label{mainthm1}
	The smooth functions $\left\{\theta_{u, p}\right\}$ on $\mathcal{M}_{m, n}^s$ are defined as follows:
	\begin{equation}\label{theta}
		\theta_{u, p}= \begin{cases}\frac{1}{2 \pi \mathrm{i}} \frac{1}{(p+1) !} \frac{s}{i+s} \int_{\Gamma} \zeta^{\frac{i}{a}} \phi_{p+1} d z, & u=t_i(i \in \mathbb{Z}-\{-s\}), \\ \frac{1}{2 \pi \mathrm{i}} \frac{s}{m} \int_{\Gamma} \frac{a^p}{p !}\left(\log \frac{\zeta^{\frac{m}{s}}}{a}+c_p\right) d z, & u=t_{-s}, \\ -\frac{\Gamma\left(1-\frac{i}{m}\right)}{\Gamma\left(2+p-\frac{i}{m}\right)} \operatorname{Res}_{z=\infty} a^{1+p-\frac{i}{m}} d z, & u=h_i(1 \leq i \leq m-1), \\ \frac{\Gamma\left(1-\frac{i}{n}\right)}{\Gamma\left(2+p-\frac{i}{n}\right)} \operatorname{Res}_{z=\varphi} \hat{a}^{1+p-\frac{i}{n}} d z, & u=\hat{h}_i(0 \leq i \leq n-1), \\ \frac{1}{2 \pi \mathrm{i}} \frac{n}{m} \int_{\Gamma} \frac{\hat{a}^p}{p !}\left(\log \left(\zeta^{\frac{m}{s}} \hat{a}^{\frac{m}{n}}\right)-\frac{m}{n} c_p\right) d z-\frac{1}{2 \pi \mathrm{i}} \frac{n}{m} \int_{\Gamma} \frac{a^p}{p !}\left(\log \frac{\zeta^{\frac{m}{s}}}{a}+c_p\right) d z, & u=\hat{h}_n,\end{cases}
	\end{equation}
where
$$
\zeta(z)=a(z)-\hat{a}(z), \quad \phi_p(z)=a^p(z)-\hat{a}^p(z), \quad c_0=0, \quad c_p=\sum_{s=1}^p \frac{1}{s} .
$$
These functions satisfy the equalities \eqref{princon1}-\eqref{princon3} associated with $\mathcal{M}_{m, n}^s$, with the constant nilpotent matrix $R=R_1$ given by
$$
\left(R_1\right)_v^u= \begin{cases}1-\frac{s}{m}, & u \in\{t_0, \hat{h}_0\},\  v=t_{-s}, \\ \frac{n}{m}+1, & u=\hat{h}_0,\  v=\hat{h}_n, \\ \frac{n}{m}-\frac{n}{s}, & u=t_0, \ v=\hat{h}_n, \\ 0, & \text { for other cases. }\end{cases}
$$
Hence, the Hamiltionian densities $\left\{\theta_{u, p}\right\}$ give rise to an infinite family of tau-symmetric commutative flows on the loop space of $\mathcal{M}_{m, n}^s$ via the equations \eqref{prin}.
\end{theorem}

A general version of the Whitham hierarchy was investigated by Krichever [18], and now we focus on the case of meromorphic functions with a fixed pole at infinity and a movable pole. More precisely, let us consider two meromophic functions of $z$ on the Riemann sphere of the form
$$
\lambda(z)=z+\sum_{i \geq 1} v_i(x)(z-\varphi(x))^{-i}, \quad \hat{\lambda}(z)=\sum_{i \geq-1} \hat{v}_i(x)(z-\varphi(x))^i,
$$
with $x \in S^1$ being a loop parameter. In this paper by the (special) Whitham hierarchy we mean the following system of evolutionary equations:
\begin{align}
	& \frac{\partial \lambda(z)}{\partial s_k}=[\left(\lambda(z)^k\right)_{+}, \lambda(z)], \quad \frac{\partial \hat{\lambda}(z)}{\partial s_k}=[\left(\lambda(z)^k\right)_{+}, \hat{\lambda}(z)], \label{whi1}\\
	& \frac{\partial \lambda(z)}{\partial \hat{s}_k}=[-(\hat{\lambda}(z)^k)_{-}, \lambda(z)], \quad \frac{\partial \hat{\lambda}(z)}{\partial \hat{s}_k}=[-(\hat{\lambda}(z)^k)_{-}, \hat{\lambda}(z)],\quad k=1,2,\cdots,\label{whi2}
\end{align}
and
\begin{equation}\label{whi3}
	\frac{\partial \lambda(z)}{\partial \hat{s}_0}=[\log (z-\varphi(x)), \lambda(z)], \quad \frac{\partial \hat{\lambda}(z)}{\partial \hat{s}_0}=[\log (z-\varphi(x)), \lambda \hat{(z)}],
\end{equation}
where the Lie bracket $[\ ,\ ]$ reads
\begin{equation}\label{bracket}
	[f, g]:=\frac{\partial f}{\partial z} \frac{\partial g}{\partial x}-\frac{\partial g}{\partial z} \frac{\partial f}{\partial x}.
\end{equation}
\begin{theorem}\label{mainthm2}
 The principal hierarchy of $\mathcal{M}_{m, n}^s$, as formulated in Theorem \ref{mainthm1}, is an extension of the Whitham hierarchy \eqref{whi1}-\eqref{whi3}.
\end{theorem}

We anticipate that our results will contribute to the understanding of geometric structure underlying the Whitham hierarchy.

The article is organized as follows: In Section 2, we will review fundamental concepts concerning the infinite-dimensional Frobenius manifold $\mathcal{M}_{m, n}^s$ and provide the explicit form of the Levi-Civita connection of $\eta$ associated with $\mathcal{M}_{m, n}^s$. Section 3 furnishes the proof of Theorem \ref{mainthm1}. In Section 4, we investigate the relationship between the principal hierarchy of $\mathcal{M}_{m, n}^s$ and the Whitham hierarchy. Lastly, concluding remarks are presented in Section 5.
\section{Overview of the infinite-dimensional Frobenius Manifold $\mathcal{M}_{m, n}^s$}
In this section, we begin by recalling some fundamental aspects of the infinite-dimensional Frobenius manifold $\mathcal{M}_{m, n}^s$. Subsequently, we present the explicit form of the Levi-Civita connection of $\eta$ associated with $\mathcal{M}_{m, n}^s$.
\subsection{The manifold $\mathcal{M}_{m, n}^s$ as a bundle on the space of parametrized simple curve}
For any $\varphi$ within the unit disk $B:=\{z \in \mathbb{C}||z |<1\}$, we consider two sets of holomorphic functions defined on disks of the Riemann sphere $\mathbb{C} \cup\{\infty\}$ as follows:
$$
\begin{aligned}
	& \mathcal{H}_{\varphi}^{-}=\left\{f(z)=\sum_{i \geq 0} f_i(z-\varphi)^{-i} \mid f \text { holomorphic on }|z|>1-\epsilon \text { for some } \epsilon>0\right\}, \\
	& \mathcal{H}_{\varphi}^{+}=\left\{f(z)=\sum_{i \geq 0} f_i(z-\varphi)^i \mid f \text { holomorphic on }|z|<1+\epsilon \text { for some } \epsilon>0\right\} .
\end{aligned}
$$

Let $m$ and $n$ be arbitrary positive integers. We define the set $\tilde{\mathcal{M}}_{m, n}$ as the union of the following terms:
$$
\tilde{\mathcal{M}}_{m, n}=\bigcup_{\varphi \in B}\left(\left(z^m+(z-\varphi)^{m-2} \mathcal{H}_{\varphi}^{-}\right) \times(z-\varphi)^{-n} \mathcal{H}_{\varphi}^{+}\right) .
$$
The elements in $\tilde{\mathcal{M}}_{m, n}$ can be expressed as Laurent series of the form:
$$
\vec{a}=\left(z^m+\sum_{i \leq m-2} a_i(z-\varphi)^i, \sum_{i \geq-n} \hat{a}_i(z-\varphi)^i\right) .
$$
For any point $\vec{a}=(a(z), \hat{a}(z)) \in \tilde{\mathcal{M}}_{m, n}$, we introduce the functions $\zeta(z)$ and $\ell(z)$ as follows:
\begin{equation}
	\zeta(z)=a(z)-\hat{a}(z), \quad \ell(z)=a(z)_{+}+\hat{a}(z)_{-}\label{zetaell}
\end{equation}
where the subscripts ' $\pm$ ' mean to take the nonnegative and the negative powers in $(z-\varphi)$. Note that $\zeta(z)$ is defined on a neighborhood of the unit circle $\Gamma:=\{|z|=1, z \in \mathbb{C}\}$, and $\ell(z)$ is holomorphic on $\mathbb{P}^{1} \backslash\{\varphi, \infty\}$.

\begin{definition}\label{deffm}
	The manifold $\mathcal{M}_{m, n}^s$, as a subset of $\tilde{\mathcal{M}}_{m, n}$, is defined by the following conditions for points $\vec{a}=(a(z), \hat{a}(z)) \in \tilde{\mathcal{M}}_{m, n}$:
	\begin{enumerate}
		\item the coefficient $\hat{a}_{-n} \neq 0$;
		\item the derivatives $\zeta^{\prime}(z)$ and $\ell^{\prime}(z)$, where $\zeta(z)$ and $\ell(z)$ are defined in \eqref{zetaell}, do not vanish on $\Gamma$;
		\item there exists a holomorphic function $w(z)$ defined in a neighborhood of $\Gamma$, with a winding number around the origin equal to $1$, such that $w(z)^s=\zeta(z)$ and it maps $\Gamma$ to a simple, smooth curve $\Gamma^{\prime}$ encircling the origin;
		\item the winding number of the functions $a(z)$ and $\hat{a}(z)$ around the origin along $\Gamma$ is $m$ and $n$, respectively.
	\end{enumerate}

\end{definition}

Let's describe the tangent and cotangent bundles of $\mathcal{M}_{m, n}^s$ using Laurent series. At each point $\vec{a}=(a(z), \hat{a}(z)) \in \mathcal{M}_{m, n}^s$, a vector $\partial$ in the tangent space $T_{\vec{a}} \mathcal{M}_{m, n}^s$ can be identified with its action $(\partial a(z), \partial \hat{a}(z))$. Therefore, the tangent space can be represented as follows:
\begin{equation}\label{tans}
T_{\vec{a}} \mathcal{M}_{m, n}^s=(z-\varphi)^{m-2} \mathcal{H}_{\varphi}^{-} \times(z-\varphi)^{-n-1} \mathcal{H}_{\varphi}^{+}.
\end{equation}
Accordingly, the cotangent space at $\vec{a}=(a(z), \hat{a}(z))$ is represented as the dual space of $T_{\vec{a}} \mathcal{M}_{m, n}^s$, denoted as:
\begin{equation}\label{cotans}
T_{\vec{a}}^* \mathcal{M}_{m, n}^s=(z-\varphi)^{-m+1} \mathcal{H}_{\varphi}^{+} \times(z-\varphi)^n \mathcal{H}_{\varphi}^{-} .
\end{equation}
The pairing between $(\omega(z), \hat{\omega}(z)) \in T_{\vec{a}}^* \mathcal{M}_{m, n}$ and $(\xi(z), \hat{\xi}(z)) \in T_{\vec{a}} \mathcal{M}_{m, n}$ is defined by:
\begin{equation}\label{pair}
\langle\vec{\omega}, \vec{\xi}\rangle:=\frac{1}{2 \pi \mathrm{i}} \int_{\Gamma}(\omega(z) \xi(z)+\hat{\omega}(z) \hat{\xi}(z)) d z .
\end{equation}
The representation of the tangent and cotangent spaces described above proves to be convenient and natural. In the following we will freely use both representations of vectors and covectors, often without specifying which one we are using, as it will typically be evident from the context.
\subsection{The metric}
In [13], the invariant metric on $\mathcal{M}_{m, n}^s$ is defined by
\begin{equation}
\langle\eta^{-1} \cdot \vec{\xi}_1, \vec{\xi}_2\rangle=\langle\vec{\xi}_1, \vec{\xi}_2\rangle_\eta, \quad  \vec{\xi}_1, \vec{\xi}_2 \in T_{\vec{a}} \mathcal{M}_{m, n}^s,
\end{equation}
where
\begin{equation*}
\eta: T_{\vec{a}}^* \mathcal{M}_{m, n}^s \rightarrow T_{\vec{a}} \mathcal{M}_{m, n}^s
\end{equation*}
is a linear bijection. For any $(\omega(z), \hat{\omega}(z)) \in T_{\vec{a}}^* \mathcal{M}_{m, n}^s$, the explicit form of the map $\eta$ is given by:
$$
\begin{aligned}
	\eta \cdot \vec{\omega}= & \left(a^{\prime}(z)[\omega(z)+\hat{\omega}(z)]_{-}-\left[\omega(z) a^{\prime}(z)+\hat{\omega}(z) \hat{a}^{\prime}(z)\right]_{-},\right. \\
	& \left.-\hat{a}^{\prime}(z)[\omega(z)+\hat{\omega}(z)]_{+}+\left[\omega(z) a^{\prime}(z)+\hat{\omega}(z) \hat{a}^{\prime}(z)\right]_{+}\right).
\end{aligned}
$$
The inverse map $\eta^{-1}$ is described by the following expressions: 
\begin{equation}\label{cov1}
\omega(z)_{+}=-\left(\frac{\xi(z)-\hat{\xi}(z)}{\zeta^{\prime}(z)}\right)_{+}, \quad \omega(z)_{-}=\left(\frac{\xi(z)}{a^{\prime}(z)}-\left(\frac{\xi(z)-\xi \hat{(z)}}{\zeta^{\prime}(z)}\right)_{-}\right)_{\geq-m+1},
\end{equation}
and
\begin{equation}\label{cov2}
\hat{\omega}(z)_{-}=\left(\frac{\xi(z)-\hat{\xi}(z)}{\zeta^{\prime}(z)}\right)_{-}, \quad \hat{\omega}_{+}=\left(-\frac{\hat{\xi}(z)}{\hat{a}^{\prime}(z)}+\left(\frac{\xi(z)-\hat{\xi}(z)}{\zeta^{\prime}(z)}\right)_{+}\right)_{\leq n},
\end{equation}
where $(\omega(z), \hat{\omega}(z))=\eta^{-1} \cdot(\xi(z), \hat{\xi}(z))$. In this context, we denote $(f(z))_{\geq m}=\sum_{s=m}^{\infty} b_s(z-\varphi)^s$ for a Laurent series $f(z)=\sum_{s=-\infty}^{+\infty} b_s(z-\varphi)^s$, and the expressions $\frac{1}{a^{\prime}(z)}$ and $\frac{1}{\hat{a}^{\prime}(z)}$ correspond to the formal
inverses of the Laurent series $a^{\prime}(z)$ and $\hat{a}^{\prime}(z)$, respectively. By using the formula \eqref{zetaell}, the invariant metric $\langle\ ,\ \rangle_\eta$ can be expressed as follows:
\begin{equation}\label{metric}
\left\langle\partial_1, \partial_2\right\rangle_\eta=-\frac{1}{2 \pi \mathrm{i}} \int_{\Gamma} \frac{\partial_1 \zeta(z) \cdot \partial_2 \zeta(z)}{\zeta^{\prime}(z)} d z-\operatorname{Res}_{z=\infty} \frac{\partial_1 \ell(z) \cdot \partial_2 \ell(z)}{\ell^{\prime}(z)} d z-\operatorname{Res}_{z=\varphi} \frac{\partial_1 \ell(z) \cdot \partial_2 \ell(z)}{\ell^{\prime}(z)} d z
\end{equation}
for any $\partial_1, \partial_2 \in T_{\vec{a}} \mathcal{M}_{m, n}^s$.
\subsection{Flat coordinates}
In accordance with condition (3) in defintion \ref{deffm}, we can consider the inverse function of $w(z)=$ $\zeta(z)^{\frac{1}{s}}$ defined in a neighborhood of $\Gamma^{\prime}$. This function can be represented by the Laurent expansion:
\begin{equation}\label{texp}
z(w)=\sum_{i \in \mathbb{Z}} t_i w^i, \quad z \in \Gamma,
\end{equation}
and the coefficients $t_i$ are determined by the following expression:
$$
t_i=\frac{1}{2 \pi \mathrm{i}} \int_{\Gamma'} z(w) w^{-i-1} d w, \quad i \in \mathbb{Z} .
$$
In addition, we introduce two functions defined in punctured neighborhoods of $\infty$ and $\varphi$ respectively. These functions are given by:
$$
\begin{aligned}
	& \chi(z):=\ell(z)^{\frac{1}{m}}=z+\chi_1 z^{-1}+\chi_2 z^{-2}+\cdots, \quad z \rightarrow \infty, \\
	& \hat{\chi}(z):=\ell(z)^{\frac{1}{n}}=\hat{a}_{-n}^{\frac{1}{n}}(z-\varphi)^{-1}+\hat{\chi}_0+\hat{\chi}_1(z-\varphi)+\cdots, \quad z \rightarrow \varphi .
\end{aligned}
$$
The inverse functions of $\chi(z)$ and $\hat{\chi}(z)$ can be expressed as follows:
\begin{align}
	& z(\chi)=\chi+h_1 \chi^{-1}+h_2 \chi^{-2}+\cdots+h_{m-1} \chi^{-m+1}+\cdots, \quad z \rightarrow \infty ,\label{hexp}\\
	& z(\hat{\chi})=\hat{h}_0+\hat{h}_1 \hat{\chi}^{-1}+\hat{h}_2 \hat{\chi}^{-2}+\cdots+\hat{h}_n \hat{\chi}^{-n}+\cdots, \quad z \rightarrow \varphi .\label{hhexp}
\end{align}
The variables $\mathbf{t} \cup \mathbf{h} \cup \hat{\mathbf{h}}=\left\{t_i\right\}_{i \in \mathbb{Z}} \cup\left\{h_j\right\}_{1 \leq j \leq m-1} \cup\{\hat{h}_k\}_{0 \leq k \leq n}$ determine $\zeta(z)$ and $\ell(z)$, and one has
$$
\begin{aligned}
	& \frac{\partial \zeta(z)}{\partial t_i}=-\zeta(z)^{\frac{1}{s}} \zeta^{\prime}(z), \quad \frac{\partial \ell(z)}{\partial t_i}=0, \\
	& \frac{\partial \zeta(z)}{\partial h_j}=0, \quad \frac{\partial \ell(z)}{\partial h_j}=\left(\ell^{\prime}(z) \chi(z)^{-j}\right)_{+}, \\
	& \frac{\partial \zeta(z)}{\partial \hat{h}_k}=0, \quad \frac{\partial \ell(z)}{\partial \hat{h}_k}=-\left(\ell^{\prime}(z) \hat{\chi}(z)^{-k}\right)_{-},
\end{aligned}
$$
which imply
$$
\begin{gathered}
	\left\langle\frac{\partial}{\partial t_{i_1}}, \frac{\partial}{\partial t_{i_2}}\right\rangle_\eta=-s \delta_{-s, i_1+i_2}, \\
	\left\langle\frac{\partial}{\partial h_{j_1}}, \frac{\partial}{\partial h_{j_2}}\right\rangle_\eta=m \delta_{m, j_1+j_2}, \\
	\left\langle\frac{\partial}{\partial \hat{h}_{k_1}}, \frac{\partial}{\partial \hat{h}_{k_2}}\right\rangle_\eta=n \delta_{n, k_1+k_2}, \\
	\left\langle\frac{\partial}{\partial t_i}, \frac{\partial}{\partial h_j}\right\rangle_\eta=\left\langle\frac{\partial}{\partial t_i}, \frac{\partial}{\partial \hat{h}_k}\right\rangle_\eta=\left\langle\frac{\partial}{\partial h_j}, \frac{\partial}{\partial \hat{h}_k}\right\rangle_\eta=0,
\end{gathered}
$$
where $i, i_1, i_2 \in \mathbb{Z}, j, j_1, j_2 \in\{1,2, \cdots, m-1\}$, and $k, k_1, k_2 \in\{0,1,2, \cdots, n\}$. Please see [13] for details.
\subsection{The Levi-Civita connection}
Let us now proceed to derive the explicit formula for the Levi-Civita connection associated with the metric $\langle\ ,\ \rangle_\eta$. Given vector fields $\partial_1$ and $\partial_2$ on $\mathcal{M}_{m, n}^s$, we consider the derivative of $\partial_2$ along $\partial_1$ that has the form:
\begin{equation}
	\begin{aligned}\label{conn}
		\left(\nabla_{\partial_1} \partial_2\right) \cdot \vec{a}= & \left(\partial_1 \partial_2 a(z)-\left[\left(\frac{\partial_1 \zeta(z) \partial_2 \zeta(z)}{\zeta^{\prime}(z)}\right)_{-}^{\prime}+\left(\frac{\partial_1 \ell(z) \partial_2 \ell(z)}{\ell^{\prime}(z)}\right)_{\infty,+}^{\prime}+\left(\frac{\partial_1 \ell(z) \partial_2 \ell(z)}{\ell^{\prime}(z)}\right)_{\varphi,-}^{\prime}\right]\right., \\
		& \ \ \left.\partial_1 \partial_2 \hat{a}(z)+\left(\frac{\partial_1 \zeta(z) \partial_2 \zeta(z)}{\zeta^{\prime}(z)}\right)_{+}^{\prime}-\left(\frac{\partial_1 \ell(z) \partial_2 \ell(z)}{\ell^{\prime}(z)}\right)_{\infty,+}^{\prime}-\left(\frac{\partial_1 \ell(z) \partial_2 \ell(z)}{\ell^{\prime}(z)}\right)_{\varphi,-}^{\prime}\right) .
	\end{aligned}
\end{equation}
Here, $(f(z))_{\varphi,-}$ denotes the negative part of the Laurent expansion of the rational function $f(z)$ at $\varphi$.
\begin{theorem}
	The connection $\nabla$ defined by \eqref{conn} is torsion-free and compatible with the metric $\langle\ ,\ \rangle_\eta$.
\end{theorem}
\begin{proof}
	The torsion-free property of the connection $\nabla$ is equivalent to the identity:
	$$
	\left(\nabla_{\partial_1} \partial_2\right) \cdot \vec{a}-\left(\nabla_{\partial_2} \partial_1\right) \cdot \vec{a}=\partial_1 \partial_2 \vec{a}-\partial_2 \partial_1 \vec{a},
	$$
	which can be readily observed from formula \eqref{conn}.
	
	The condition for the connection $\nabla$ to be compatible with the metric $\langle\ ,\ \rangle_\eta$ is expressed as follows:
	\begin{equation}\label{compa}
		\nabla_{\partial_1}\left\langle\partial_2, \partial_3\right\rangle_\eta=\left\langle\nabla_{\partial_1} \partial_2, \partial_3\right\rangle_\eta+\left\langle\partial_2, \nabla_{\partial_1} \partial_3\right\rangle_\eta
	\end{equation}
for any vector field $\partial_1, \partial_2$ and $\partial_3$ on $\mathcal{M}_{m, n}^s$. To verify formula \eqref{compa}, we can apply the Leibniz rule and formula \eqref{metric}, which lead to
\begin{equation}
	\begin{aligned}\label{compfor1}
		& \nabla_{\partial_1}\left\langle\partial_2, \partial_3\right\rangle_\eta \\
		= & -\frac{1}{2 \pi \mathrm{i}} \int_{\Gamma}\left( \frac{\partial_1 \partial_2 \zeta(z) \cdot \partial_3 \zeta(z)}{\zeta^{\prime}(z)}+\frac{\partial_2 \zeta(z) \cdot \partial_1 \partial_3 \zeta(z)}{\zeta^{\prime}(z)} - \frac{\partial_2 \zeta(z) \cdot \partial_3 \zeta(z) \cdot \partial_1 \zeta^{\prime}(z)}{\left(\zeta^{\prime}(z)\right)^2}\right) d z \\
		& -\operatorname{Res}_{z=\infty}\left( \frac{\partial_1 \partial_2 \ell(z) \cdot \partial_3 \ell(z)}{\ell^{\prime}(z)} + \frac{\partial_2 \ell(z) \cdot \partial_1 \partial_3 \ell(z)}{\ell^{\prime}(z)} - \frac{\partial_2 \ell(z) \cdot \partial_3 \ell(z) \cdot \partial_1 \ell^{\prime}(z)}{\left(\ell^{\prime}(z)\right)^2} \right)d z \\
		& -\operatorname{Res}_{z=\varphi} \left(\frac{\partial_1 \partial_2 \ell(z) \cdot \partial_3 \ell(z)}{\ell^{\prime}(z)} + \frac{\partial_2 \ell(z) \cdot \partial_1 \partial_3 \ell(z)}{\ell^{\prime}(z)}- \frac{\partial_2 \ell(z) \cdot \partial_3 \ell(z) \cdot \partial_1 \ell^{\prime}(z)}{\left(\ell^{\prime}(z)\right)^2} \right)d z .
	\end{aligned}
\end{equation}
By using the commutativity of $\partial_\nu$ and $\partial_z$ for $\nu=1,2$ and applying integration by parts, we can derive the following results:
\begin{equation}
	\begin{aligned}\label{compfor2}
		& \frac{1}{2 \pi \mathrm{i}} \int_{\Gamma} \frac{\partial_2 \zeta(z) \cdot \partial_3 \zeta(z) \cdot \partial_1 \zeta^{\prime}(z)}{\left(\zeta^{\prime}(z)\right)^2} d z \\
		= & -\frac{1}{2 \pi \mathrm{i}} \int_{\Gamma}\left(\frac{\partial_2 \zeta(z)}{\zeta^{\prime}(z)}\right)^{\prime} \frac{\partial_3 \zeta(z) \cdot \partial_1 \zeta(z)}{\zeta^{\prime}(z)} d z-\frac{1}{2 \pi \mathrm{i}} \int_{\Gamma}\left(\frac{\partial_3 \zeta(z)}{\zeta^{\prime}(z)}\right)^{\prime} \frac{\partial_2 \zeta(z) \cdot \partial_1 \zeta(z)}{\zeta^{\prime}(z)} d z \\
		= & \frac{1}{2 \pi \mathrm{i}} \int_{\Gamma} \frac{\partial_2 \zeta(z)}{\zeta^{\prime}(z)}\left(\frac{\partial_3 \zeta(z) \cdot \partial_1 \zeta(z)}{\zeta^{\prime}(z)}\right)^{\prime} d z+\frac{1}{2 \pi \mathrm{i}} \int_{\Gamma} \frac{\partial_3 \zeta(z)}{\zeta^{\prime}(z)}\left(\frac{\partial_2 \zeta(z) \cdot \partial_1 \zeta(z)}{\zeta^{\prime}(z)}\right)^{\prime} d z,
	\end{aligned}
\end{equation}
and
\begin{equation}
	\begin{aligned}\label{compfor3}
		& \operatorname{Res}_{z=\infty} \frac{\partial_2 \ell(z) \cdot \partial_3 \ell(z) \cdot \partial_1 \ell^{\prime}(z)}{\left(\ell^{\prime}(z)\right)^2} d z \\
		= & \operatorname{Res}_{z=\infty} \frac{\partial_2 \ell(z)}{\ell^{\prime}(z)}\left(\frac{\partial_3 \ell(z) \cdot \partial_1 \ell(z)}{\left(\ell^{\prime}(z)\right)^2}\right)^{\prime} d z+\operatorname{Res}_{z=\infty} \frac{\partial_3 \ell(z)}{\ell^{\prime}(z)}\left(\frac{\partial_2 \ell(z) \cdot \partial_1 \ell(z)}{\left(\ell^{\prime}(z)\right)^2}\right)^{\prime} d z,
	\end{aligned}
\end{equation}
and
\begin{equation}
	\begin{aligned}\label{compfor4}
		& \operatorname{Res}_{z=\varphi} \frac{\partial_2 \ell(z) \cdot \partial_3 \ell(z) \cdot \partial_1 \ell^{\prime}(z)}{\left(\ell^{\prime}(z)\right)^2} d z \\
		= & \operatorname{Res}_{z=\varphi} \frac{\partial_2 \ell(z)}{\ell^{\prime}(z)}\left(\frac{\partial_3 \ell(z) \cdot \partial_1 \ell(z)}{\left(\ell^{\prime}(z)\right)^2}\right)^{\prime} d z+\operatorname{Res}_{z=\varphi} \frac{\partial_3 \ell(z)}{\ell^{\prime}(z)}\left(\frac{\partial_2 \ell(z) \cdot \partial_1 \ell(z)}{\left(\ell^{\prime}(z)\right)^2}\right)^{\prime} d z .
	\end{aligned}
\end{equation}
Substituting \eqref{compfor2}-\eqref{compfor4} into \eqref{compfor1}, one has
\begin{equation}
\begin{aligned}\label{compfor5}
	& \nabla_{\partial_1}\left\langle\partial_2, \partial_3\right\rangle_\eta \\
	=&-\frac{1}{2 \pi \mathrm{i}} \int_{\Gamma}\left( \frac{\partial_1 \partial_2 \zeta(z) \cdot \partial_3 \zeta(z)}{\zeta^{\prime}(z)} +  \frac{\partial_2 \zeta(z) \cdot \partial_1 \partial_3 \zeta(z)}{\zeta^{\prime}(z)} -  \frac{\partial_2 \zeta(z)}{\zeta^{\prime}(z)}\left(\frac{\partial_3 \zeta(z) \cdot \partial_1 \zeta(z)}{\zeta^{\prime}(z)}\right)^{\prime} - \frac{\partial_3 \zeta(z)}{\zeta^{\prime}(z)}\left(\frac{\partial_2 \zeta(z) \cdot \partial_1 \zeta(z)}{\zeta^{\prime}(z)}\right)^{\prime} \right) dz\\
	&-\operatorname{Res}_{z=\infty} \left(\frac{\partial_1 \partial_2 \ell(z) \cdot \partial_3 \ell(z)}{\ell^{\prime}(z)} + \frac{\partial_2 \ell(z) \cdot \partial_1 \partial_3 \ell(z)}{\ell^{\prime}(z)} - \frac{\partial_2 \ell(z)}{\ell^{\prime}(z)}\left(\frac{\partial_3 \ell(z) \cdot \partial_1 \ell(z)}{\left(\ell^{\prime}(z)\right)^2}\right)^{\prime}- \frac{\partial_3 \ell(z)}{\ell^{\prime}(z)}\left(\frac{\partial_2 \ell(z) \cdot \partial_1 \ell(z)}{\left(\ell^{\prime}(z)\right)^2}\right)^{\prime}\right) d z \\
	&-\operatorname{Res}_{z=\varphi}\left( \frac{\partial_1 \partial_2 \ell(z) \cdot \partial_3 \ell(z)}{\ell^{\prime}(z)} + \frac{\partial_2 \ell(z) \cdot \partial_1 \partial_3 \ell(z)}{\ell^{\prime}(z)} - \frac{\partial_2 \ell(z)}{\ell^{\prime}(z)}\left(\frac{\partial_3 \ell(z) \cdot \partial_1 \ell(z)}{\left(\ell^{\prime}(z)\right)^2}\right)^{\prime} - \frac{\partial_3 \ell(z)}{\ell^{\prime}(z)}\left(\frac{\partial_2 \ell(z) \cdot \partial_1 \ell(z)}{\left(\ell^{\prime}(z)\right)^2}\right)^{\prime}\right) d z .
\end{aligned}
\end{equation}
On the other hand, formula \eqref{conn} can be equivalently expressed as:
\begin{equation}\label{compfor6}
	\left(\nabla_{\partial_1} \partial_2\right) \cdot \zeta(z)=\partial_1 \partial_2 \zeta(z)-\left(\frac{\partial_1 \zeta(z) \partial_2 \zeta(z)}{\zeta^{\prime}(z)}\right)^{\prime},
\end{equation}
\begin{equation}\label{compfor7}
	\left(\nabla_{\partial_1} \partial_2\right) \cdot \ell(z)=\partial_1 \partial_2 \ell(z)-\left(\frac{\partial_1 \ell(z) \partial_2 \ell(z)}{\ell^{\prime}(z)}\right)_{\infty,+}^{\prime}-\left(\frac{\partial_1 \ell(z) \partial_2 \ell(z)}{\ell^{\prime}(z)}\right)_{\varphi,-}^{\prime} .
\end{equation}
By combining \eqref{metric}, \eqref{compfor5}, \eqref{compfor6}, and \eqref{compfor7}, we can deduce that formula \eqref{compa} holds. The theorem is proved.
\end{proof}

\begin{corollary}
	The explicit expression for the covariant derivative of a one-form field $\vec{\omega}=(\omega(z), \hat{\omega}(z))$ along a vector field $\partial$ on $\mathcal{M}_{m, n}^s$ is given by:
	\begin{equation}
		\begin{aligned}\label{dulconn}
			\nabla_{\partial} \vec{\omega}= & \left(\left(\partial \omega(z)-\left(\omega^{\prime}(z)_{+}-\hat{\omega}^{\prime}(z)_{-}\right) \frac{\partial \zeta(z)}{\zeta^{\prime}(z)}-\left(\omega^{\prime}(z)_{-}+\hat{\omega}^{\prime}(z)_{-}\right) \frac{\partial \ell(z)}{a^{\prime}(z)}\right)_{\geq-m+1}\right. \\
			& \left.\left(\partial \hat{\omega}(z)+\left(\omega^{\prime}(z)_{+}-\hat{\omega}^{\prime}(z)_{-}\right) \frac{\partial \zeta(z)}{\zeta^{\prime}(z)}-\left(\omega^{\prime}(z)_{+}+\hat{\omega}^{\prime}(z)_{+}\right) \frac{\partial \ell(z)}{\hat{a}^{\prime}(z)}\right)_{\leq n}\right) .
		\end{aligned}
	\end{equation}

\end{corollary}
\begin{proof}
	Let us verify the formula:
	\begin{equation}
		\partial_1\left\langle\vec{\omega}, \partial_2\right\rangle=\left\langle\nabla_{\partial_1} \vec{\omega}, \partial_2\right\rangle+\left\langle\vec{\omega}, \nabla_{\partial_1} \partial_2\right\rangle .
	\end{equation}
By the use of \eqref{pair} and \eqref{conn}, one has:
$$
\partial_1\left\langle\vec{\omega}, \partial_2\right\rangle=\int\left(\partial_1 \omega(z) \partial_2 a(z)+\omega(z) \partial_1 \partial_2 a(z)+\partial_1 \hat{\omega}(z) \partial_2 \hat{a}(z)+\hat{\omega}(z) \partial_1 \partial_2 \hat{a}(z)\right) d z,
$$
and
$$
\begin{aligned}
	& \left\langle\vec{\omega}, \nabla_{\partial_1} \partial_2\right\rangle_\eta \\
	= & \int_{\Gamma}\left(\omega(z) \partial_1 \partial_2 a(z)-\omega(z)\left(\frac{\partial_1 \zeta(z) \partial_2 \zeta(z)}{\zeta^{\prime}(z)}\right)_{-}^{\prime}-\omega(z)\left(\frac{\partial_1 \ell(z) \partial_2 \ell(z)}{\ell^{\prime}(z)}\right)_{\infty,+}^{\prime}-\omega(z)\left(\frac{\partial_1 \ell(z) \partial_2 \ell(z)}{\ell^{\prime}(z)}\right)_{\varphi,-}^{\prime}\right) d z \\
	& +\int_{\Gamma}\left(\hat{\omega}(z) \partial_1 \partial_2 \hat{a}(z)+\hat{\omega}(z)\left(\frac{\partial_1 \zeta(z) \partial_2 \zeta(z)}{\zeta^{\prime}(z)}\right)_{-}^{\prime}-\hat{\omega}(z)\left(\frac{\partial_1 \ell(z) \partial_2 \ell(z)}{\ell^{\prime}(z)}\right)_{\infty,+}^{\prime}-\hat{\omega}(z)\left(\frac{\partial_1 \ell(z) \partial_2 \ell(z)}{\ell^{\prime}(z)}\right)_{\varphi,-}^{\prime}\right) d z \\
	= & \int_{\Gamma}\left(\omega(z) \partial_1 \partial_2 a(z)+\hat{\omega}(z) \partial_1 \partial_2 \hat{a}(z)+\omega^{\prime}(z)_{+}\left(\frac{\partial_1 \zeta(z) \partial_2 \zeta(z)}{\zeta^{\prime}(z)}\right)-\hat{\omega}^{\prime}(z)_{-}\left(\frac{\partial_1 \zeta(z) \partial_2 \zeta(z)}{\zeta^{\prime}(z)}\right)\right) d z \\
	& -\operatorname{Res}_{\infty}\left(\omega^{\prime}(z)_{-}\left(\frac{\partial_1 \ell(z) \partial_2 \ell(z)}{a^{\prime}(z)}\right)+\hat{\omega}^{\prime}(z)_{-}\left(\frac{\partial_1 \ell(z) \partial_2 \ell(z)}{a^{\prime}(z)}\right)\right) d z \\
	& +\operatorname{Res}_{\varphi}\left(\omega^{\prime}(z)_{+}\left(\frac{\partial_1 \ell(z) \partial_2 \ell(z)}{\hat{a}^{\prime}(z)}\right)+\hat{\omega}^{\prime}(z)_{+}\left(\frac{\partial_1 \ell(z) \partial_2 \ell(z)}{\hat{a}^{\prime}(z)}\right)\right) d z .
\end{aligned}
$$
Hence
$$
\begin{aligned}
	& \partial_1\left\langle\vec{\omega}, \partial_2\right\rangle-\left\langle\vec{\omega}, \nabla_{\partial_1} \partial_2\right\rangle \\
	= & \int_{\Gamma}\left(\partial_1 \omega(z)-\omega^{\prime}(z)_{+} \frac{\partial_1 \zeta(z)}{\zeta^{\prime}(z)}+\hat{\omega}^{\prime}(z)_{-} \frac{\partial_1 \zeta(z)}{\zeta^{\prime}(z)}\right) \partial_2 a(z) d z+\operatorname{Res}_{\infty}\left(\left(\omega^{\prime}(z)_{-}+\hat{\omega}^{\prime}(z)-\right) \frac{\partial_1 \ell(z)}{a^{\prime}(z)}\right) \partial_2 a(z) d z \\
	& +\int_{\Gamma}\left(\partial_1 \hat{\omega}(z)+\omega^{\prime}(z)_{+} \frac{\partial_1 \zeta(z)}{\zeta^{\prime}(z)}-\hat{\omega}^{\prime}(z)_{-} \frac{\partial_1 \zeta(z)}{\zeta^{\prime}(z)}\right) \partial_2 \hat{a}(z) d z-\operatorname{Res}_{\varphi}\left(\left(\omega^{\prime}(z)_{+}+\hat{\omega}^{\prime}(z)_{+}\right) \frac{\partial_1 \ell(z)}{\hat{a}^{\prime}(z)}\right) \partial_2 \hat{a}(z) d z \\
	= & \left\langle\nabla_{\partial_1} \vec{\omega}, \partial_2\right\rangle
\end{aligned}
$$
for $\nabla_{\partial_1} \vec{\omega}$ defined by \eqref{dulconn}. The corollary is proved.
\end{proof}
\subsection{The associative product and Euler vector field}
The Frobenius manifolds are endowed with an associative commutative product on each tangent space. For any $\vec{\omega}_1=\left(\omega_1(z), \hat{\omega}_1(z)\right), \vec{\omega}_2=\left(\omega_2(z), \hat{\omega}_2(z)\right) \in T_{\vec{a}}^* \mathcal{M}_{m, n}^s$, the product $*$ on cotangent space $T^* \mathcal{M}_{m, n}^s$ can be represented as follows:
\begin{equation}
	\begin{aligned}\label{corpro}
		\vec{\omega}_1 * \vec{\omega}_2= & \left(\left[\omega_2(z)\left(\omega_1(z) a^{\prime}(z)\right)_{+}-\omega_1(z)\left(\omega_2(z) a^{\prime}(z)\right)_{-}-\omega_2(z)\left(\hat{\omega}_1(z) \hat{a}^{\prime}(z)\right)_{-}-\omega_1(z)\left(\hat{\omega}_2(z) \hat{a}^{\prime}(z)\right)_{-}\right]_{\geq-m+1},\right. \\
		& {\left.\left[\hat{\omega}_2(z)\left(\hat{\omega}_1(z) \hat{a}^{\prime}(z)\right)_{+}-\hat{\omega}_1(z)\left(\hat{\omega}_2(z) \hat{a}^{\prime}(z)\right)_{-}+\hat{\omega}_1(z)\left(\omega_2(z) a^{\prime}(z)\right)_{+}+\hat{\omega}_2(z)\left(\omega_1(z) a^{\prime}(z)\right)_{+}\right]_{\leq n}\right) . }
	\end{aligned}
\end{equation}
The bijection $\eta: T^* \mathcal{M}_{m, n}^s \rightarrow T \mathcal{M}_{m, n}^s$ induces an associative product $\circ$ on $T \mathcal{M}_{m, n}^s$, satisfing
\begin{equation*}
\vec{\xi}_1 \circ \vec{\xi}_2:=\eta\left(\eta^{-1}\left(\vec{\xi}_1\right) * \eta^{-1}\left(\vec{\xi}_2\right)\right), \quad \xi_1, \xi_2 \in T \mathcal{M}_{m, n}^s .
\end{equation*}
The unity vectors for the product $*$ and $\circ$ are given by
\begin{equation}
\vec{e}^*=\left(\frac{1}{m}(z-\varphi)^{-m+1}, 0\right) \in T_{\vec{a}}^* \mathcal{M}_{m, n}^s,
\end{equation}
and
\begin{equation}
\vec{e}= \begin{cases}\frac{1}{m} \frac{\partial}{\partial h_{m-1}}, & m \geq 2, \\ \frac{\partial}{\partial t_{s-1}}+\frac{\partial}{\partial \hat{h}_0}, & m=1,\end{cases}
\end{equation}
respectively.

Given a tangent vector $\vec{\xi} \in T_{\vec{a}} \mathcal{M}_{m, n}^s$, the multiplication by $\eta^{-1}(\vec{\xi})$ induces a linear map $C_{\vec{\xi}}$ : $T_{\vec{a}}^* \mathcal{M}_{m, n}^s \rightarrow T_{\vec{a}}^* \mathcal{M}_{m, n}^s$ that sends $\vec{\omega}$ to $\eta^{-1}(\vec{\xi}) * \vec{\omega}$. We will now proceed to deduce the explicit form of this map.
\begin{lemma}
	Let $\partial$ be a vector field and $\vec{\omega}=(\omega(z), \hat{\omega}(z))$ be a one-form field on $\mathcal{M}_{m, n}^s$. The explicit form of the operator $C_{\partial}$ defined above is given by
	\begin{equation}
		\begin{aligned}\label{dulcorpro}
			C_{\partial}(\vec{\omega})= & \left(\left(\left(\frac{\partial a(z)}{a^{\prime}(z)}-\frac{\partial \zeta(z)}{\zeta^{\prime}(z)}\right) a^{\prime}(z) \omega(z)-\left(\frac{\partial a(z)}{a^{\prime}(z)}-\frac{\partial \zeta(z)}{\zeta^{\prime}(z)}\right)\left(\omega(z) a^{\prime}(z)+\hat{\omega}(z) \hat{a}^{\prime}(z)\right)_{-}\right)_{\geq-m+1},\right. \\
			& \left.\left(\left(\frac{\partial \hat{a}(z)}{\hat{a}^{\prime}(z)}-\frac{\partial \zeta(z)}{\zeta^{\prime}(z)}\right) \hat{a}^{\prime}(z) \hat{\omega}(z)-\left(\frac{\partial \hat{a}(z)}{\hat{a}^{\prime}(z)}-\frac{\partial \zeta(z)}{\zeta^{\prime}(z)}\right)\left(\omega(z) a^{\prime}(z)+\hat{\omega}(z) \hat{a}^{\prime}(z)\right)_{+}\right)_{\leq n}\right) .
		\end{aligned}
	\end{equation}
\end{lemma}
\begin{proof}
	Let us deduce formula \eqref{dulcorpro} from \eqref{cov1},\eqref{cov2} and \eqref{corpro}. Note that the right-hand side of \eqref{corpro} remains unchanged when an element from $(z-\varphi)^{-m} \mathcal{H}_{\varphi}^{-} \times(z-\varphi)^n \mathcal{H}_{\varphi}^{+}$ is added to $\left(\omega_1(z), \hat{\omega}_1(z)\right)$. According to \eqref{cov1} and \eqref{cov2}, we substitute
	$$
	\omega_1(z)=\left(\frac{\xi_1(z)}{a^{\prime}(z)}-\frac{\xi_1(z)-\hat{\xi}_1(z)}{\zeta^{\prime}(z)}\right), \quad \hat{\omega}_1(z)=\left(\frac{\xi_1(z)-\hat{\xi}_1(z)}{\zeta^{\prime}(z)}-\frac{\hat{\xi}_1(z)}{\hat{a}^{\prime}(z)}\right)
	$$
	into the right-hand side of \eqref{corpro}. By using the equality:
	$$
	\omega_1(z) a^{\prime}(z)+\hat{\omega}_1(z) \hat{a}^{\prime}(z)=0,
	$$
	we obtain
	$$
	\begin{aligned}
		C_{\partial_1} \vec{\omega}_2= & \left(\left(\omega_1(z) \omega_2(z) a^{\prime}(z)-\omega_1(z)\left(\omega_2(z) a^{\prime}(z)+\hat{\omega}_2(z) \hat{a}^{\prime}(z)\right)_{-}\right)_{\geq-m+1}\right. \\
		& \left.\left(-\hat{\omega}_1(z) \hat{\omega}_2(z) \hat{a}^{\prime}(z)+\hat{\omega}_1(z)\left(\omega_2(z) a^{\prime}(z)+\hat{\omega}_2(z) \hat{a}^{\prime}(z)\right)_{+}\right)_{\leq n}\right) \\
		= & \left(\left(\left(\frac{\partial_1 a(z)}{a^{\prime}(z)}-\frac{\partial_1 \zeta(z)}{\zeta^{\prime}(z)}\right) a^{\prime}(z) \omega_2(z)-\left(\frac{\partial_1 a(z)}{a^{\prime}(z)}-\frac{\partial_1 \zeta(z)}{\zeta^{\prime}(z)}\right)\left(\omega_2(z) a^{\prime}(z)+\hat{\omega}_2(z) \hat{a}^{\prime}(z)\right)_{-}\right)_{\geq-m+1},\right. \\
		& \left.\left(\left(\frac{\partial_1 \hat{a}(z)}{\hat{a}^{\prime}(z)}-\frac{\partial_1 \zeta(z)}{\zeta^{\prime}(z)}\right) \hat{a}^{\prime}(z) \hat{\omega}_2(z)-\left(\frac{\partial_1 \hat{a}(z)}{\hat{a}^{\prime}(z)}-\frac{\partial_1 \zeta(z)}{\zeta^{\prime}(z)}\right)\left(\omega_2(z) a^{\prime}(z)+\hat{\omega}_2(z) \hat{a}^{\prime}(z)\right)_{+}\right)_{\leq n}\right) .
	\end{aligned}
	$$
	The lemma is proved.
\end{proof}

The Euler vector field $E$ on $\mathcal{M}_{m, n}^s$ is given by the expression:
$$
E=\left(a(z)-\frac{z a^{\prime}(z)}{m}, \hat{a}(z)-\frac{z \hat{a}^{\prime}(z)}{m}\right),
$$
or, equivalently
$$
E=\sum_{i \in \mathbb{Z}}\left(\frac{1}{m}-\frac{i}{s}\right) t_i \frac{\partial}{\partial t_i}+\sum_{i \geq 1}^{m-1}\left(\frac{1}{m}+\frac{i}{m}\right) h_i \frac{\partial}{\partial h_i}+\sum_{i \geq 0}^n\left(\frac{1}{m}+\frac{i}{n}\right) \hat{h}_i \frac{\partial}{\partial \hat{h}_i}.
$$
 Hence the diagonal matrix $\mu=\operatorname{diag}\left(\mu_1, \cdots, \mu_n\right)$ defined by the formula \eqref{bigmu} takes the following form:
\begin{equation}\label{infmu}
	\mu_u= \begin{cases}\frac{i}{s}+\frac{1}{2}, & u=t_i(i \in \mathbb{Z}), \\ \frac{1}{2}-\frac{i}{m}, & u=h_i(1 \leq i \leq m-1), \\ \frac{1}{2}-\frac{i}{n}, & u=\hat{h}_i(0 \leq i \leq n).\end{cases}
\end{equation}
The intersection form $g$ on $\mathcal{M}_{m,n}^{s}$ can be constructed by using the equation \eqref{inter}. For further details, please refer to [13].

\section{Proof of theorem \ref{mainthm1}}
In this section, we will demonstrate that the functions $\theta_{\alpha, p}(t)$ on $\mathcal{M}_{m, n}^s$, given by \eqref{theta}, satisfy the equalities \eqref{princon1}-\eqref{princon3}, where the explicity form of the  constant matrix $\mu$ and $R$ are given by the formula \eqref{infmu} and 
$$
\left(R\right)_v^u= \begin{cases}1-\frac{s}{m}, & u \in\{t_0, \hat{h}_0\}, v=t_{-s}, \\ \frac{n}{m}+1, & u=\hat{h}_0, v=\hat{h}_n, \\ \frac{n}{m}-\frac{n}{s}, & u=t_0, v=\hat{h}_n, \\ 0, & \text { otherwise, }\end{cases}
$$
respectively.

The equalities \eqref{princon3} can be readily obtained by using the equations \eqref{texp}-\eqref{hhexp}. To verify the equalities \eqref{princon2}, we introduce the operator $\mathcal{E}=E+\frac{1}{m} z \frac{\partial}{\partial z}$. Then   one has the following formulas:
\begin{equation}
	\operatorname{Lie}_{\mathcal{E}} \zeta(z)=\zeta(z), \quad \operatorname{Lie}_{\mathcal{E}} a(z)=a(z), \quad \operatorname{Lie}_{\mathcal{E}} \hat{a}(z)=\hat{a}(z), \quad \operatorname{Lie}_{\mathcal{E}} \phi_p(z)=p \phi_p(z).\label{homo}
\end{equation}
From the equality \eqref{homo} and 
$$
\int_{\Gamma} \operatorname{Lie}_{\mathcal{E}} f(a, \hat{a}) d z=\operatorname{Lie}_E \int_{\Gamma} f(a, \hat{a}) d z-\frac{1}{m} \int_{\Gamma} f(a, \hat{a}) d z,
$$
it follows that
$$
\operatorname{Lie}_E \theta_{u, p}(t)= \begin{cases}\left(p+1+\frac{i}{s}+\frac{1}{m}\right) \theta_{u, p}(t), & u=t_i(i \in \mathbb{Z}-\{-s\}), \\ \left(p+\frac{1}{m}\right) \theta_{u, p}(t)+\left(1-\frac{s}{m}\right)\left(\theta_{t_0, p-1}(t)+\theta_{\hat{h}_0, p-1}(t)\right), & u=t_{-s}, \\ \left(1+p-\frac{i}{m}+\frac{1}{m}\right) \theta_{u, p}(t), & u=h_i(1 \leq i \leq m-1), \\ \left(1+p-\frac{i}{n}+\frac{1}{m}\right) \theta_{u, p}(t), & u=\hat{h}_i(0 \leq i \leq n-1), \\ \left(p+\frac{1}{m}\right) \theta_{u, p}(t)+\left(\frac{n}{m}+1\right) \theta_{\hat{h}_0, p-1}(t)+\left(\frac{n}{m}-\frac{n}{s}\right) \theta_{t_0, p-1}(t), & u=\hat{h}_n .\end{cases}
$$
Hence we obtain the equalities \eqref{princon2}.

For the set of functions $\left\{\theta_{t_i, p} \mid i \in \mathbb{Z}-\{-s\}, p \in \mathbb{N}\right\}$, let us verify the equalities \eqref{princon1}. Denote
$$
Q_{t_i, p}(a, \hat{a})=\frac{1}{(p+1) !} \frac{s}{i+s} \zeta^{\frac{i}{s}} \phi_{p+1},
$$
then one has $\theta_{t_i, p}(t)=\frac{1}{2 \pi \mathrm{i}} \int_{\Gamma} Q_{t_i, p}(a, \hat{a}) d z$ and
$$
d \theta_{t_i, p}(t)=\left(\left(\frac{\partial Q_{t_i, p}(a, \hat{a})}{\partial a}\right)_{\geq-m+1},\left(\frac{\partial Q_{t_i, p}(a, \hat{a})}{\partial \hat{a}}\right)_{\leq n}\right) .
$$
We can verify that the functions $\left\{Q_{t_i, p}(a, \hat{a})\right\}$ satisfy
$$
Q_{t_i, p}(a, \hat{a})=\frac{\partial Q_{t_i, p+1}(a, \hat{a})}{\partial a}+\frac{\partial Q_{t_i, p+1}(a, \hat{a})}{\partial \hat{a}},
$$
which leads to
\begin{equation}\label{proman1}
\frac{\partial^2 Q_{t_i, p+1}(a, \hat{a})}{\partial a \partial a}+\frac{\partial^2 Q_{t_i, p+1}(a, \hat{a})}{\partial a \partial \hat{a}}=\frac{\partial Q_{t_i, p}(a, \hat{a})}{\partial a}
\end{equation}
and
\begin{equation}\label{proman11}
\frac{\partial^2 Q_{t_i, p+1}(a, \hat{a})}{\partial a \partial \hat{a}}+\frac{\partial^2 Q_{t_i, p+1}(\hat{a}, \hat{a})}{\partial a \partial \hat{a}}=\frac{\partial Q_{t_i, p}(a, \hat{a})}{\partial \hat{a}} .
\end{equation}
We need to show that for any vector field $\partial=(\partial a, \partial \hat{a})$ on $\mathcal{M}_{m, n}^s$, the following relation holds:
\begin{equation}
\nabla_{\partial} d \theta_{t_i, p+1}(t)=C_{\partial} d \theta_{t_i, p}(t),\label{proman22}
\end{equation}
which is equivalent to the equalities \eqref{princon1}. Substituting $\vec{\omega}=d \theta_{t_i, p+1}$ into \eqref{dulconn}, one has
\begin{equation}
	\begin{aligned}\label{proman2}
		& \nabla_{\partial} d \theta_{t_i, p+1} \\
		= & \left(\left(\frac{\partial^2 Q_{t_i, p+1}}{\partial a \partial a} \partial a+\frac{\partial^2 Q_{t_i, p+1}}{\partial a \partial \hat{a}} \partial \hat{a}-\left(\left(\frac{\partial Q_{t_i, p+1}}{\partial a}\right)_{+}^{\prime}-\left(\frac{\partial Q_{t_i, p+1}}{\partial \hat{a}}\right)_{-}^{\prime}\right) \frac{\partial \zeta}{\zeta^{\prime}}\right.\right. \\
		& \left.-\left(\left(\frac{\partial Q_{t_i, p+1}}{\partial a}\right)_{-}^{\prime}+\left(\frac{\partial Q_{t_i, p+1}}{\partial \hat{a}}\right)_{-}^{\prime}\right) \frac{\partial \ell}{a^{\prime}}\right)_{\geq-m+1}, \\
		& \left(\frac{\partial^2 Q_{t_i, p+1}}{\partial a \partial \hat{a}} \partial a+\frac{\partial^2 Q_{t_i, p+1}}{\partial \hat{a} \partial \hat{a}} \partial \hat{a}+\left(\left(\frac{\partial Q_{t_i, p+1}}{\partial a}\right)_{+}^{\prime}-\left(\frac{\partial Q_{t_i, p+1}}{\partial \hat{a}}\right)_{-}^{\prime}\right) \frac{\partial \zeta}{\zeta^{\prime}}\right. \\
		& \left.\left.-\left(\left(\frac{\partial Q_{t_i, p+1}}{\partial a}\right)_{+}^{\prime}+\left(\frac{\partial Q_{t_i, p+1}}{\partial \hat{a}}\right)_{+}^{\prime}\right) \frac{\partial \ell}{\hat{a}^{\prime}}\right)_{\leq n}\right) .
	\end{aligned}
\end{equation}
By using the expressions
$$
\partial a=\partial \zeta_{-}+\partial \ell, \quad \partial \hat{a}=-\partial \zeta_{+}+\partial \ell, \quad \partial \zeta=\partial \zeta_{+}+\partial \zeta_{-},
$$
we can extract the terms that involve $\partial \zeta_{+}, \partial \zeta_{-}$and $\partial \ell$ individually in \eqref{proman2} as follows:
$$
\begin{aligned}
\partial \zeta_{+}: \quad 	& \left(\left(-\frac{\partial^2 Q_{t_i, p+1}}{\partial a \partial \hat{a}} \partial \zeta_{+}-\left(\left(\frac{\partial Q_{t_i, p+1}}{\partial a}\right)_{+}^{\prime}-\left(\frac{\partial Q_{t_i, p+1}}{\partial \hat{a}}\right)_{-}^{\prime}\right) \frac{\partial \zeta_{+}}{\zeta^{\prime}}\right)_{\geq-m+1},\right. \\
	& \left.\left(-\frac{\partial^2 Q_{t_i, p+1}}{\partial \hat{a} \partial \hat{a}} \partial \zeta_{+}+\left(\left(\frac{\partial Q_{t_i, p+1}}{\partial a}\right)_{+}^{\prime}-\left(\frac{\partial Q_{t_i, p+1}}{\partial \hat{a}}\right)_{-}^{\prime}\right) \frac{\partial \zeta_{+}}{\zeta^{\prime}}\right)_{\leq n}\right) ; \\
\partial \zeta_{-}:	\quad& \left(\left(\frac{\partial^2 Q_{t_i, p+1}}{\partial a \partial a} \partial \zeta_{-}-\left(\left(\frac{\partial Q_{t_i, p+1}}{\partial a}\right)_{+}^{\prime}-\left(\frac{\partial Q_{t_i, p+1}}{\partial \hat{a}}\right)_{-}^{\prime}\right) \frac{\partial \zeta_{-}}{\zeta^{\prime}}\right)_{\geq-m+1},\right. \\
	& \left.\left(\frac{\partial^2 Q_{t_i, p+1}}{\partial a \partial \hat{a}} \partial \zeta_{-}+\left(\left(\frac{\partial Q_{t_i, p+1}}{\partial a}\right)_{+}^{\prime}-\left(\frac{\partial Q_{t_i, p+1}}{\partial \hat{a}}\right)_{-}^{\prime}\right) \frac{\partial \zeta_{-}}{\zeta^{\prime}}\right)_{\leq n}\right) ; \\
\partial \ell:\quad	& \left(\left(\frac{\partial^2 Q_{t_i, p+1}}{\partial a \partial a} \partial \ell+\frac{\partial^2 Q_{t_i, p+1}}{\partial a \partial \hat{a}} \partial \ell-\left(\left(\frac{\partial Q_{t_i, p+1}}{\partial a}\right)_{-}^{\prime}+\left(\frac{\partial Q_{t_i, p+1}}{\partial \hat{a}}\right)_{-}^{\prime}\right) \frac{\partial \ell}{a^{\prime}}\right)_{\geq-m+1},\right. \\
	& \left.\left(\frac{\partial^2 Q_{t_i, p+1}}{\partial a \partial \hat{a}} \partial \ell+\frac{\partial^2 Q_{t_i, p+1}}{\partial \hat{a} \partial \hat{a}} \partial \ell-\left(\left(\frac{\partial Q_{t_i, p+1}}{\partial a}\right)_{+}^{\prime}+\left(\frac{\partial Q_{t_i, p+1}}{\partial \hat{a}}\right)_{+}^{\prime}\right) \frac{\partial \ell}{\hat{a}^{\prime}}\right)_{\leq n}\right) . 
\end{aligned}
$$
For the right hand side of the equality \eqref{proman22}, by the use of the formula \eqref{dulcorpro}, we can express $C_{\partial}\left(d \theta_{t_i, p}\right)$ in the form:
$$
\begin{aligned}
	& C_{\partial}\left(d \theta_{t_i, p}\right) \\
	= & \left(\left(\left(\frac{\partial a}{a^{\prime}}-\frac{\partial \zeta}{\zeta^{\prime}}\right) a^{\prime} \frac{\partial Q_{t_i, p}}{\partial a}-\left(\frac{\partial a}{a^{\prime}}-\frac{\partial \zeta}{\zeta^{\prime}}\right)\left(\frac{\partial Q_{t_i, p}}{\partial a} a^{\prime}+\frac{\partial Q_{t_i, p}}{\partial \hat{a}} \hat{a}^{\prime}\right)_{-}\right)_{\geq-m+1},\right. \\
	& \left.\left(\left(\frac{\partial \hat{a}}{\hat{a}^{\prime}}-\frac{\partial \zeta}{\zeta^{\prime}}\right) \hat{a}^{\prime} \frac{\partial Q_{t_i, p}}{\partial \hat{a}}-\left(\frac{\partial \hat{a}}{\hat{a}^{\prime}}-\frac{\partial \zeta}{\zeta^{\prime}}\right)\left(\frac{\partial Q_{t_i, p}}{\partial a} a^{\prime}+\frac{\partial Q_{t_i, p}}{\partial \hat{a}} \hat{a}^{\prime}\right)_{+}\right)_{\leq n}\right) .
\end{aligned}
$$
The terms in this expression that involve $\partial \zeta_{+}, \partial \zeta_{-}$, and $\partial \ell$ individually are given by:
$$
\begin{aligned}
	\partial \zeta_{+}:  \quad& \left(-\frac{a^{\prime}}{\zeta^{\prime}} \frac{\partial Q_{t_i, p}}{\partial a} \partial \zeta_{+}+\frac{\partial \zeta_{+}}{\zeta^{\prime}}\left(\frac{\partial Q_{t_i, p}}{\partial a} a^{\prime}+\frac{\partial Q_{t_i, p}}{\partial \hat{a}} \hat{a}^{\prime}\right)_{-}\right)_{\geq-m+1}, \\
	& \left.\left(-\frac{\partial Q_{t_i, p}}{\partial \hat{a}} \partial \zeta_{+}-\frac{\hat{a}^{\prime}}{\zeta^{\prime}} \frac{\partial Q_{t_i, p}}{\partial \hat{a}} \partial \zeta_{+}+\frac{\partial \zeta_{+}}{\zeta^{\prime}}\left(\frac{\partial Q_{t_i, p}}{\partial a} a^{\prime}+\frac{\partial Q_{t_i, p}}{\partial \hat{a}} \hat{a}^{\prime}\right)_{+}\right)_{\leq n}\right) \\
	\partial \zeta_{-}:  \quad& \left(\left(\frac{\partial Q_{t_i, p}}{\partial a} \partial \zeta_{-}-\frac{a^{\prime}}{\zeta^{\prime}} \frac{\partial Q_{t_i, p}}{\partial a} \partial \zeta_{-}+\frac{\partial \zeta_{-}}{\zeta^{\prime}}\left(\frac{\partial Q_{t_i, p}}{\partial a} a^{\prime}+\frac{\partial Q_{t_i, p}}{\partial \hat{a}} \hat{a}^{\prime}\right)_{-}\right)_{\geq-m+1},\right. \\
	& \left.\left(-\frac{\hat{a}^{\prime}}{\zeta^{\prime}} \frac{\partial Q_{t_i, p}}{\partial \hat{a}} \partial \zeta_{-}+\frac{\partial \zeta_{-}}{\zeta^{\prime}}\left(\frac{\partial Q_{t_i, p}}{\partial a} a^{\prime}+\frac{\partial Q_{t_i, p}}{\partial \hat{a}} \hat{a}^{\prime}\right)_{+}\right)_{\leq n}\right) \\
	\partial \ell:  \quad& \left(\left(\frac{\partial Q_{t_i, p}}{\partial a} \partial \ell-\frac{\partial \ell}{a^{\prime}}\left(\frac{\partial Q_{t_i, p}}{\partial a} a^{\prime}+\frac{\partial Q_{t_i, p}}{\partial \hat{a}} \hat{a}^{\prime}\right)_{-}\right)_{\geq-m+1},\right. \\
	& \left.\left(\frac{\partial Q_{t_i, p}}{\partial \hat{a}} \partial \ell-\frac{\partial \ell}{\hat{a}^{\prime}}\left(\frac{\partial Q_{t_i, p}}{\partial a} a^{\prime}+\frac{\partial Q_{t_i, p}}{\partial \hat{a}} \hat{a}^{\prime}\right)_{+}\right)_{\leq n}\right) .
\end{aligned}
$$
By comparing the extracted terms involving $\partial \zeta_{+}, \partial \zeta_{-}$, and $\partial \ell$, and using the identities \eqref{proman1} and \eqref{proman11}, we can establish the equality \eqref{proman22}. The remaining cases can be verified by using a similar procedure. Theorem \ref{mainthm1} is proved.
\section{Relationships between the principal hierarchy and other integrable hierarchies}
In this section, we commence by revisiting fundamental facets of the dispersionless extended KP hierarchy and its bihamiltonian structure.  Subsequently, we present the proof of theorem \ref{mainthm2}. Finally, we establish the connection between the principal hierarchy of $\mathcal{M}_{m, n}^s$ and the principal hierarchy of a finite-dimensional Frobenius manifold.

\subsection{The dispersionless extended KP hierarchy}
Let us review some fundamental aspects of the dispersionless extended KP hierarchy \eqref{whi1}-\eqref{whi2}, as investigated in [17]. We consider the loop space, denoted as $\mathcal{L} \mathcal{M}_{m, n}^s$, consisting of smooth maps from the unit circle $S^1$ to the manifold $\mathcal{M}_{m, n}^s$. A point in the loop space $\mathcal{L} \mathcal{M}_{m, n}^s$ can be expressed as:
$$
\vec{a}=(a, \hat{a})=\left(z^m+\sum_{i \leq m-2} a_i(x)(z-\varphi(x))^i, \sum_{i \geq-n} \hat{a}_i(x)(z-\varphi(x))^i\right) .
$$
At this point, the tangent space $T_{\vec{a}} \mathcal{L M}_{m, n}^s$ and the cotangent space $T_{\vec{a}}^* \mathcal{L} \mathcal{M}_{m, n}^s$ of the loop space take a similar form as \eqref{tans} and \eqref{cotans}, respectively.

On the loop space $\mathcal{L} \mathcal{M}_{m, n}^s$, we introduce a ring of formal differential polynomials denoted as
$$
\mathcal{A}:=C^{\infty}(\mathbf{a})\left[\left[\partial_x \mathbf{a}, \partial_x{ }^2 \mathbf{a}, \ldots\right]\right],
$$
where
$$
\mathbf{a}(x)=\left(\varphi(x), a_{m-2}(x), a_{m-3}(x), \ldots, \hat{a}_{-n}(x), \hat{a}_{-n+1}(x), \ldots\right) .
$$
We now consider the quotient space $\mathcal{F}:=\mathcal{A} / \partial_x \mathcal{A}$, whose elements are called local functionals and can be expressed in the following form:
$$
F=\int f\left(\mathbf{a}, \partial_x \mathbf{a}, \partial_x{ }^2 \mathbf{a}, \ldots\right) d x \in \mathcal{F}, \quad f \in \mathcal{A}.
$$
Given a local functional $F \in \mathcal{F}$, its variational gradient at $\vec{a}$ refers to a cotangent vector $d F \in$ $T_{\vec{a}}^* \mathcal{L} \mathcal{M}_{m, n}^s$ that satisfies
$$
\delta F=\int\langle d F, \delta \vec{a}\rangle d x,
$$
where the nondegenerate pairing $\langle\ ,\ \rangle$ is defined in \eqref{pair}.
\begin{proposition}
 $([17])$ For any positive integers $m$ and $n$, a bihamiltonian structure exists on $\mathcal{L} \mathcal{M}_{m, n}^s$ defined by two compatible Poisson brackets:
	$$
	\{F, H\}_\nu=\int\left\langle d F, \mathcal{P}_\nu d H\right\rangle d x, \quad \nu=1,2 .
	$$
	Here the Poisson tensors $\mathcal{P}_\nu: T_{\vec{a}}^* \mathcal{L} \mathcal{M}_{m, n} \rightarrow T_{\vec{a}} \mathcal{L} \mathcal{M}_{m, n}$ read
	\begin{align}
		\mathcal{P}_1 \cdot \vec{\omega}= & \left([\omega, a]_{-}+[\hat{\omega}, \hat{a}]_{-}-\left[\omega_{-}+\hat{\omega}_{-}, a\right]-[\omega, a]_{+}-[\hat{\omega}, \hat{a}]_{+}+\left[\omega_{+}+\hat{\omega}_{+}, \hat{a}\right]\right) ,\label{ham1}\\
		\mathcal{P}_2 \cdot \vec{\omega}= &\left(\left([\omega, a]_{-}+[\hat{\omega}, \hat{a}]_{-}\right) a-\left[(\omega a+\hat{\omega} \hat{a})_{-}, a\right]-\sigma a^{\prime}\right. , \nonumber\\ 
		& \left.-\left([\omega, a]_{+}+[\hat{\omega}, \hat{a}]_{+}\right) \hat{a}+\left[(\omega a+\hat{\omega} \hat{a})_{+}, \hat{a}\right]_{-}-\sigma \hat{a}^{\prime}\right)\label{ham2}
	\end{align}
	with the Lie bracket $[\ ,\ ]$ defined by  \eqref{bracket}, and
	$$
	\sigma=\frac{1}{m} \operatorname{Res}_{z=\varphi}([\omega, a]+[\hat{\omega}, \hat{a}]) d z .
	$$
	Furthermore, let
	$$
	a=\lambda(z)^m, \quad \hat{a}=\hat{\lambda}(z)^n,
	$$then the hierarchy \eqref{whi1}-\eqref{whi2} can be formulated in a bihamiltonian recursive form as follows:
	$$
	\frac{\partial F}{\partial s_k}=\left\{F, H_{k+m}\right\}_1=\left\{F, H_k\right\}_2, \quad \frac{\partial F}{\partial \hat{s}_k}=\{F, \hat{H}_{k+n}\}_1=\{F, \hat{H}_k\}_2,
	$$
	where the Hamiltonian functionals are given by\begin{equation}\label{hamitan}
		H_k=-\frac{m}{k} \int \operatorname{Res}_{z=\infty} \lambda(z)^k d x, \quad \hat{H}_k=\frac{n}{k} \int \operatorname{Res}_{z=\varphi} \hat{\lambda}(z)^k d x
	\end{equation}with $k=1,2,3, \ldots$.
\end{proposition}
\begin{corollary}
$([13])$ The dispersionless bihamiltonian structure defined by flat pencil $(\eta, g)$ on the loop space $\mathcal{L}\mathcal{M}_{m, n}^n$ coincide with $\left(\mathcal{P}_1, \mathcal{P}_2\right)$ given by \eqref{ham1} and \eqref{ham2}, respectively.
\end{corollary}

\subsection{Proof of theroem \ref{mainthm2}}
By comparing the explicit forms of the Hamiltonians $H_k, \hat{H}_k$ given by \eqref{hamitan} and the expression for $\int \theta_{u, p}, d x$ provided in \eqref{theta}, we can immediately deduce the following relationships:
$$
H_k= \begin{cases}\frac{\Gamma\left(1+p-\frac{i}{m}\right)}{\Gamma\left(1-\frac{i}{m}\right)} \int \theta_{h_i, p} d x, & k=m(p+1)-i, \quad i=1, \cdots, m-1, \\ \Gamma(p+1)\left(\int \theta_{t_0, p} d x+\int \theta_{\hat{h}_0, p} d x\right), & k=m(p+1),\end{cases}
$$
and
$$
\hat{H}_k=\frac{\Gamma\left(1+p-\frac{i}{n}\right)}{\Gamma\left(1-\frac{i}{n}\right)} \int \theta_{\hat{h}_i, p} d x, \quad k=n(p+1)-i, \quad i=0, \cdots, n-1 .
$$
Considering the logarithmic flows defined in equation \eqref{whi3}, we observe that, using the formula $\left[\frac{\partial Q_{u, p}}{\partial a}, a\right]+\left[\frac{\partial Q_{u, p}}{\partial \hat{a}}, \hat{a}\right]=0$, the Hamiltonian vector fields of $\theta_{u, p}=\frac{1}{2 \pi \mathrm{i}} \int_{\Gamma} Q_{u, p} d z$ under the Hamiltonian structure $\mathcal{P}_1$ can be expressed as follows:
$$
\begin{aligned}
	\mathcal{P}_1 \cdot d \theta_{u, p} & =\left(-\left[\left(\frac{\partial Q_{u, p}}{\partial a}+\frac{\partial Q_{u, p}}{\partial \hat{a}}\right)_{-}, a\right],\left[\left(\frac{\partial Q_{u, p}}{\partial a}+\frac{\partial Q_{u, p}}{\partial \hat{a}}\right)_{+}, \hat{a}\right]\right) \\
	& =\left(-\left[\left(Q_{u, p-1}\right)_{-}, a\right],\left[\left(Q_{u, p-1}\right)_{+}, \hat{a}\right]\right).
\end{aligned}
$$
Then
$$
\begin{aligned}
	\mathcal{P}_1 \cdot d \theta_{\hat{h}_n, 1} & =\left(-\left[\left(Q_{\hat{h}_n, 0}\right)_{-}, a\right],\left[\left(Q_{\hat{h}_n, 0}\right)_{+}, \hat{a}\right]\right) \\
	& =\left(-\left[\left(n \log \hat{a}^{\frac{1}{n}} \zeta^{\frac{1}{m}}+n \log \frac{a^{\frac{1}{m}}}{\zeta^{\frac{1}{m}}}\right)_{-}, a\right],\left[\left(n \log \hat{a}^{\frac{1}{n}} \zeta^{\frac{1}{m}}+n \log \frac{a^{\frac{1}{m}}}{\zeta^{\frac{1}{m}}}\right)_{+}, \hat{a}\right]\right) \\
	& =\left(-\left[\left(n \log \hat{a}^{\frac{1}{n}}(z-\varphi)+n \log \frac{a^{\frac{1}{m}}}{z-\varphi}\right)_{-}, a\right],\left[\left(n \log \hat{a}^{\frac{1}{n}}(z-\varphi)+n \log \frac{a^{\frac{1}{m}}}{z-\varphi}\right)_{+}, \hat{a}\right]\right) \\
	& =\left(-\left[n \log \frac{a^{\frac{1}{m}}}{z-\varphi}, a\right],\left[n \log \hat{a}^{\frac{1}{n}}(z-\varphi), \hat{a}\right]\right) \\
	& =\left(\left[n \log (z-\varphi), a\right],\left[n \log (z-\varphi), \hat{a}\right]\right) .
\end{aligned}
$$
Therefore, we can conclude that $\frac{\partial}{\partial \hat{s}_0}=\frac{1}{n} \frac{\partial}{\partial T^{h_n, 0}}$ and the principal hierarchy of $\mathcal{M}_{m, n}^s$ is an extension of the Whitham hierarchy \eqref{whi1}-\eqref{whi3}. Theroem \ref{mainthm2} is proved.

Let us give some examples of the flows $\left\{\frac{\partial}{\partial T^{u, p}}\right\}$ given by the Hamiltonian densities $\left\{\theta_{u, p}\right\}$.
\begin{example}
	The explicity form of the flow $\frac{\partial}{\partial T^{t_i, 1}}$ for $i \neq-s$ is given by:
	$$
	\begin{aligned}
		\frac{\partial a}{\partial T^{t_i, 1}} & =-\left[\left(Q_{t_i, 0}\right)_{-}, a\right] \\
		& =-\left[\left(\frac{s}{i+s} \zeta^{\frac{i+s}{s}}\right)_{-}, a\right] \\
		& =-\left[(a-\hat{a})^{\frac{i}{s}}\left(a^{\prime}-\hat{a}^{\prime}\right)\right]_{-} a_x+\left[(a-\hat{a})^{\frac{i}{s}}\left(a_x-\hat{a}_x\right)\right]_{-} a^{\prime}
	\end{aligned}
	$$and
	$$
	\frac{\partial \hat{a}}{\partial T^{t_i, 1}}=\left[(a-\hat{a})^{\frac{i}{s}}\left(a^{\prime}-\hat{a}^{\prime}\right)\right]_{+} \hat{a}_x-\left[(a-\hat{a})^{\frac{i}{s}}\left(a_x-\hat{a}_x\right)\right]_{+} \hat{a}^{\prime} .
	$$
\end{example}
\begin{example}
	The explicity form of the flow $\frac{\partial}{\partial T^{t-s, 1}}$ is given by:
	$$
	\begin{aligned}
		\frac{\partial a}{\partial T^{t_i, 1}} & =-\left[\left(Q_{t_{-s}, 0}\right)_{-}, a\right] \\
		& =-\left[\left(\frac{s}{m} \log \frac{\zeta^{\frac{m}{s}}}{a}\right)_{-}, a\right] \\
		& =-\left[\frac{\left(a^{\prime}-\hat{a}^{\prime}\right)}{(a-\hat{a})}-\frac{s a^{\prime}}{m a}\right]_{-} a_x+\left[\frac{\left(a_x-\hat{a}_x\right)}{(a-\hat{a})}-\frac{s a_x}{m a}\right]_{-} a^{\prime},
	\end{aligned}
	$$and
	$$
	\frac{\partial \hat{a}}{\partial T^{t_i, 1}}=\left[\frac{\left(a^{\prime}-\hat{a}^{\prime}\right)}{(a-\hat{a})}-\frac{s a^{\prime}}{m a}\right]_{+} \hat{a}_x-\left[\frac{\left(a_x-\hat{a}_x\right)}{(a-\hat{a})}-\frac{s a_x}{m a}\right]_{+} \hat{a}^{\prime} .
	$$
\end{example}
\subsection{Relation to the princial hierarchy of a finite-dimensional Frobenius mainfold}
To reduce $\mathcal{M}_{m, n}^s$ to a finite-dimensional Frobenius manifold, slight modifications are required in its definition.

Typically, to ensure the well-definedness of the Hamiltonian densities $\theta_{t_{-s}, p}$ and $\theta_{\hat{h}_n, p}$ given in \eqref{theta}, condition (4) in definition \ref{deffm} is imposed. However, it is possible to remove this condition by redefining $\theta_{t_{-s}, p}$ and $\theta_{\hat{h}_n, p}$ as follows:
$$
\begin{aligned}
	\theta_{t_{-s}, p}= & \frac{1}{2 \pi \mathrm{i}} \frac{s}{m} \int_{\Gamma} \frac{a^p}{p !} \log \frac{\zeta^{\frac{m}{s}}}{(z-\varphi)^m} d z-\frac{s}{m} \operatorname{Res}_{\infty} \frac{a^p}{p !}\left(\log \frac{(z-\varphi)^m}{a}+c_p\right) d z, \\
	\theta_{\hat{h}_n, p}= & \frac{1}{2 \pi \mathrm{i}} \frac{n}{m} \int_{\Gamma} \frac{\hat{a}^p}{p !} \log \frac{\zeta^{\frac{m}{s}}}{(z-\varphi)^m} d z+\frac{n}{m} \operatorname{Res}_{z=\varphi} \frac{\hat{a}^p}{p !}\left(\log \left(\hat{a}^{\frac{m}{n}}(z-\varphi)^m\right)-\frac{m}{n} c_p\right) d z \\
	& \left.-\frac{1}{2 \pi \mathrm{i}} \frac{n}{m} \int_{\Gamma} \frac{a^p}{p !} \log \frac{\zeta^{\frac{m}{s}}}{(z-\varphi)^m} d z+\frac{n}{m} \operatorname{Res}_{\infty} \frac{a^p}{p !}\left(\log \frac{(z-\varphi)^m}{a}\right)+c_p\right) d z .
\end{aligned}
$$
It should be noted that even with this modified formulation, the theorems \ref{mainthm1} and \ref{mainthm2} still hold true.

As $\zeta(z)$ approaches zero, the manifold $\mathcal{M}_{m, n}^s$ converges to the finite-dimensional manifold denoted by $M$. This manifold is defined as the set of rational functions with the form:
$$
\ell(z)=z^m+a_{m-2} z^{m-2}+\cdots+a_0+a_{-1}(z-\varphi)^{-1}+\cdots+a_{-n}(z-\varphi)^{-n},
$$
and carries a Frobenius manifold structure. For tangent vectors $\partial^{\prime}, \partial^{\prime \prime}, \partial^{\prime \prime \prime} \in T_{\ell(z)} M$, the invariant metric $\eta$ and the 3-tensor $c$ asscoiated with $M$ are given by:
$$
\begin{aligned}
	& \eta\left(\partial^{\prime}, \partial^{\prime \prime}\right)=\sum_{|\ell|<\infty} \operatorname{Res}_{d \ell=0} \frac{\partial^{\prime}(\ell(z) d z) \partial^{\prime \prime}(\ell(z) d z)}{d \ell(z)}, \\
	& c\left(\partial^{\prime}, \partial^{\prime \prime}, \partial^{\prime \prime \prime}\right):=\sum_{|\ell|<\infty} \operatorname{Res}_{d \ell=0} \frac{\partial^{\prime}(\lambda(z) d z) \partial^{\prime \prime}(\ell(z) d z) \partial^{\prime \prime \prime}(\ell(z) d z)}{d \ell(z) d z} .
\end{aligned}
$$
The multiplication structure is defined by $c\left(\partial^{\prime}, \partial^{\prime \prime}, \partial^{\prime \prime \prime}\right)=\eta\left(\partial^{\prime} \circ \partial^{\prime \prime}, \partial^{\prime \prime \prime}\right)$. The flat coodinates denoted by $\left\{t_{0, j}\right\}_{j=1}^{m-1} \cup\left\{t_{1, j}\right\}_{j=0}^n$ of the flat metric $\eta$ is given by
$$
z= \begin{cases}t_{1,0}+t_{1,1} \ell(z)^{-\frac{1}{n}}+t_{1,2} \ell(z)^{-\frac{2}{n}}+\cdots, & z \rightarrow \varphi, \\ \ell^{\frac{1}{m}}-t_{0,1} \ell(z)^{-\frac{1}{m}}-t_{0,2} \ell(z)^{-\frac{2}{m}}+\cdots, & z \rightarrow \infty .\end{cases}
$$
The principal hierarchy of $M$ is provided by the following Hamiltonian density:
$$
\begin{aligned}
	\theta_{t_{0, j} ; p} & =-\operatorname{Res}_{\infty} \frac{\Gamma\left(1-\frac{j}{m}\right)}{\Gamma\left(2+p-\frac{j}{m}\right)} \ell(z)^{\frac{(p+1) m-j}{m}} d z, \quad j=1, \cdots, m-1, \\
	\theta_{t_{1, j} ; p} & =\operatorname{Res}_{\varphi} \frac{\Gamma\left(1-\frac{j}{n}\right)}{\Gamma\left(2+p-\frac{j}{n}\right)} \ell(z)^{\frac{(p+1) n-j}{n}} d z, \quad j=0, \cdots, n-1, \\
	\theta_{t_{1, n} ; p} & =-\operatorname{Res}_{\infty}\left(\frac{n}{m} \frac{\ell(z)^p}{p !}\left(\log \frac{\ell(z)}{(z-\varphi)^m}-c_{0, p}\right)\right) d z+\operatorname{Res}_{\varphi}\left(\frac{\ell(z)^p}{p !}\left(\log (z-\varphi)^n \ell(z)-c_{1, p}\right)\right) d z,
\end{aligned}
$$
where
$$
c_{0, p}=\frac{1}{m} \sum_{s=1}^p \frac{1}{s}, \quad c_{1, p}=\frac{1}{n} \sum_{s=1}^p \frac{1}{s}, \quad c_{0,0}=c_{1,0}=0 .
$$
For more details, please refer to $[19,1]$. Under the mentioned limit, it is evident that the Frobenius manifold structure on $\mathcal{M}_{m, n}^s$ will converge to that of $M$. Moreover, the Hamiltonian densities $\left\{\theta_{h_i ; p}, \theta_{\hat{h}_j ; p}\right\}$ will converge to $\left\{\theta_{t_{0, i} ; p}, \theta_{t_{1, j} ; p}\right\}$, respectively.

\section{Conclusion}
In this article, we have constructed the principal hierarchy of the infinite-dimensional Frobenius manifold $\mathcal{M}_{m, n}^s$ via providing an explicit solution of the deformed flatness equations of $\mathcal{M}_{m, n}^s$. We have further demonstrated that this hierarchy serves as an extension of the special genus zero Whitham hierarchy.

We conjecture the existence of an extension of the general Whitham hierarchy that serves as the principal hierarchy of some infinite-dimensional Frobenius manifolds, which we intend to investigate in our future studies. Moreover, we aim to conduct a comprehensive study of the tau function associated with this hierarchy and establish its intricate relationships with the seminal works in $[18,20$, $21,22]$.

\textbf{Acknowledgements}. This research did not receive any specific grant from funding agencies in the public, commercial, or not-for-profit sectors. I am grateful to Chao-Zhong Wu and Dafeng Zuo for their helpful insights and guidance in pointing out this promising research direction.

\end{document}